\documentclass[journal]{IEEEtran}

\usepackage{amsmath,amssymb,amsfonts}
\usepackage{amsthm}
\usepackage{bm}
\usepackage{graphicx}
\usepackage{cite}
\usepackage{booktabs}
\usepackage{comment}
\usepackage{url}
\usepackage{color}
\usepackage{float}
\usepackage{stfloats}
\usepackage{textcomp}
\usepackage[caption=false,font=footnotesize]{subfig}
\usepackage[nameinlink,noabbrev]{cleveref}
\usepackage{microtype}

\usepackage{algorithm}
\usepackage{algorithmic}

\newcommand{\1}{\mathsf{1}}

\DeclareMathOperator*{\argmax}{arg\,max}
\DeclareMathOperator*{\argmin}{arg\,min}

\newcommand{\sat}[1]{\min\!\left\{\,1,\left[#1\right]_+\,\right\}}
\newcommand{\Ball}[2]{\mathbb{B}_{#1} \left(#2\right)}
\usepackage{booktabs}
\usepackage{tabularx}
\usepackage{array}
\usepackage{threeparttable}
\usepackage{ragged2e}

\newcolumntype{Y}{>{\RaggedRight\arraybackslash}X}
\newtheorem{assumption}{Assumption}
\newtheorem{definition}{Definition}
\newtheorem{theorem}{Theorem}

\newtheorem{lemma}{Lemma}
\newtheorem{problem}{Problem}
\newtheorem{remark}{Remark}
\newtheorem{corollary}{Corollary}
\crefname{theorem}{Theorem}{Theorems}
\crefname{remark}{Remark}{Remarks}
\crefname{lemma}{Lemma}{Lemmas}
\crefname{corollary}{Corollary}{Corollary}
\crefname{assumption}{Assumption}{Assumptions}
\crefname{definition}{Definition}{Definitions}
\crefname{problem}{Problem}{Problems}
\crefname{figure}{Figure}{Figures}
\begin{document}
\setlength\textfloatsep{8pt}
\setlength\intextsep{8pt}

\title{Data-Driven Synthesis of Probabilistic Controlled Invariant Sets for Linear MDPs}

\author{Kazumune Hashimoto, Shunki Kimura, Kazunobu Serizawa, Junya Ikemoto, Yulong Gao, Kai Cai%
\thanks{Kazumune Hashimoto, Shunki Kimura, Kazunobu Serizawa, Junya Ikemoto are with the Graduate School of Engineering, The University of Osaka, 2-1 Yamadaoka, Suita, Osaka 565-0871, Japan. Yulong Gao is with the Department of Electrical and Electronic Engineering, Imperial College London. Kai Cai is with the Graduate School of Informatics, Osaka Metropolitan University}%
\thanks{Corresponding author: Kazumune Hashimoto. 
This work was partially supported by JST
ACT-X JPMJAX23CK, and JSPS KAKENHI Grant Numbers 25K07794
and 22KK0155}%
}

\markboth{}{Hashimoto \MakeLowercase{\textit{et al.}}: Data-Driven Synthesis of Probabilistic Controlled Invariant Sets for Linear MDPs}

\maketitle

\begin{abstract}
We study data-driven computation of \emph{probabilistic controlled invariant sets} (PCIS) for safety-critical reinforcement learning under unknown dynamics.
Assuming a linear MDP model, we use regularized least squares and self-normalized confidence bounds to construct a conservative estimate of the states from which the system can be kept inside a prescribed safe region over an \(N\)-step horizon, together with the corresponding set-valued safe action map.
This construction is obtained through a backward recursion and can be interpreted as a conservative approximation of the \(N\)-step safety predecessor operator.
When the associated conservative-inclusion event holds, a conservative fixed point of the approximate recursion can be certified as an \((N,\epsilon)\)-PCIS with confidence at least \(\eta\).
For continuous state spaces, we introduce a lattice abstraction and a Lipschitz-based discretization error bound to obtain a tractable approximation scheme.
Finally, we use the resulting conservative fixed-point approximation as a runtime candidate PCIS  in a practical shielding architecture with iterative updates, and illustrate the approach on a numerical experiment.
\end{abstract}


\section{Introduction}\label{sec:intro}

Reinforcement learning (RL) enables an agent to learn a control policy through trial-and-error interactions with its environment.
Recent successes in games \cite{mnih2015humanlevel} have spurred growing interest in deploying RL in real-world systems such as robotics, autonomous driving, and industrial control.
However, because RL fundamentally relies on exploration, it may select actions that violate safety constraints during learning, potentially leading to severe accidents or system damage. In safety-critical domains, therefore, it is important to ensure safety not only for the final deployed policy but also throughout the learning process.
This motivates \emph{safe reinforcement learning} (safe RL), which augments the learning loop with explicit safety considerations \cite{garcia2015survey,gu2022review,brunke2022safelearning}.

Despite substantial progress, guaranteeing safety during exploration in unknown environments remains challenging.
Some approaches ensure safe learning by assuming additional mechanisms, such as a recovery controller, a reset policy, or an explicit emergency action that can steer the system back to a safe region when safety is at risk \cite{Tommaso:17,thananjeyan2021recovery}.
Although theoretically powerful, such assumptions may be unrealistic or undesirable in many physical systems.
Other approaches address safety through constrained policy optimization, safe exploration with learned constraints, or asymptotic bounds on cumulative violations \cite{achiam2017cpo,chow2018lyapunov,wachi2020cmdp,ghosh2022modelfree,Wei:2024,hashimoto:23}.
These approaches are highly relevant for safe learning under constraints, but their guarantees are typically expressed in terms of feasible policy updates, expected constraint satisfaction, or cumulative violation bounds, rather than as an external runtime shield defined by an explicit (probabilistic) safe set and updated directly from data under unknown dynamics.

A complementary line of work enforces safety online through shielding or safety filtering.
Representative examples include temporal-logic shielding, probabilistic shielding, online shielding, and approximate model-based shielding \cite{alshiekh2018safe,jansen2020probshield,konighofer2023onlineshield,goodall2023approxshield}.
In addition, predictive safety filters, reachability-based safety frameworks, Gaussian-process-based safety filters, control barrier function methods, and safe model-based RL provide related runtime intervention mechanisms \cite{wabersich2021psf,hsu2024safetyfilter,safelearning2,japtap,taylor2020cbf,felix2}.
These approaches show that online intervention can be highly effective once an appropriate safe set or safety filter is available.
This leaves a central question for safe RL: can such a probabilistic safe set be constructed and enlarged directly from data when the dynamics are unknown?

We approach this question through the lens of \emph{controlled invariance} \cite{blanchini1999a}.
Classical controlled invariant sets provide a set-theoretic characterization of safety: if the current state lies in such a set, then there exists a control action that keeps the trajectory inside the set indefinitely \cite{blanchini1999a,gilbert1991a}.
For stochastic dynamics and Markov decision processes (MDPs), we adopt the probabilistic counterpart introduced in \cite{yulong2021}, namely \emph{probabilistic controlled invariant sets} (PCIS).
Informally, a PCIS is a subset of the safety region from which the agent can choose actions so that the trajectory remains safe with high probability over a prescribed horizon.
Related reachability- and feasible-set-based perspectives have also recently been explored in RL \cite{yu2022reachability,wachi2020cmdp}.
When a PCIS is available, it naturally induces a runtime shield: at each step, the RL learner proposes an action, and the shield restricts execution to actions certified safe with respect to the current candidate PCIS.

The main challenge is to compute such a probabilistic safe set when the transition kernel is unknown and the state space may be continuous.
In this paper, we consider MDPs with unknown dynamics, and adopt the \emph{linear MDP} framework \cite{jin20a,Papini:2021}, in which the transition kernel is represented through a known feature map and unknown linear coefficients.
Recent safe RL results with linear function approximation have shown that this structure can support regret analysis under several formulations, including unknown safety constraints, instantaneous hard constraints, and cumulative constraints \cite{amani2021safe,ghosh2022modelfree,shi2023hard,Wei:2024}.
Building on this structure, and using regularized least squares together with self-normalized confidence bounds \cite{abbasi2011online}, we derive an \(\eta\)-conservative approximation of the \(N\)-step safety predecessor operator for a fixed reference set.
We then show that a conservative fixed point satisfying the corresponding conservative-inclusion event is certified as an \((N,\epsilon)\)-PCIS, and that the same backward recursion induces a stage-dependent family of set-valued safe action maps.
For continuous state spaces, we further introduce a lattice abstraction together with a Lipschitz-based discretization bound to obtain a tractable approximation scheme.

Finally, we embed the resulting safe action maps into a practical shielding architecture.
Starting from an initial PCIS, the shield filters the executed actions of an arbitrary RL learner, while the current candidate PCIS and the accompanying safe action map are updated from accumulated data.
The theoretical results in this paper certify conservative operator evaluation for a fixed set.
When the same dataset is reused to search for a conservative fixed point and to repeat shield updates online, the resulting procedure should therefore be interpreted as a practical data-driven shielding template motivated by the theory, unless an additional certification argument is available.
\begin{table*}[t]
\centering
\caption{Assumption-oriented comparison with representative approaches closest to our setting.}
\label{tab:intro_comparison}
\footnotesize
\setlength{\tabcolsep}{4pt}
\renewcommand{\arraystretch}{1.12}
\begin{threeparttable}
\begin{tabularx}{\textwidth}{@{}p{3.25cm}YYY@{}}
\toprule
Paradigm &
What must be available a-priori? &
How is safety enforced during learning? &
What is synthesized from data? \\
\midrule

Data-driven invariant-set synthesis / verification
\cite{korda2020ciset,chen2018mrci,mulagaleti2022ris,mejari2023rci,kashani2025invariant,strong2024verification,griffioen2024gp}
&
State/input constraints together with a model class, geometric set parameterization, excitation/Lipschitz assumptions, or a probabilistic dynamics prior
&
Invariant-type sets are synthesized or certified offline / in batch, often together with a controller or governor
&
Positive / robust / controlled / probabilistic invariant set, often together with a controller or governor
\\

Recovery/reset-based safe exploration
\cite{Tommaso:17,thananjeyan2021recovery}
&
Recovery policy, reset mechanism, or a learned recovery model
&
Runtime fallback to a recovery action when risk is high
&
Recovery zone, risk predictor, or fallback policy
\\

CMDP / constrained policy optimization
\cite{achiam2017cpo,chow2018lyapunov}
&
Constraint costs and a feasible (or approximately feasible) policy-update mechanism
&
Safety is incorporated directly into the RL policy update
&
Policy and value / constraint estimates
\\

Safe-region expansion / abstraction-based methods
\cite{hashimoto:23,HASHIMOTO2022110646}
&
Initial safe seed together with uncertainty-quantification, smoothness, or abstraction assumptions
&
Expand a certified safe region and optimize inside it
&
Certified safe region and sometimes a safety controller
\\

Given-model shielding / safety filtering
\cite{alshiekh2018safe,jansen2020probshield,konighofer2023onlineshield,wabersich2021psf,goodall2023approxshield}
&
Safety-relevant model fragment or a given safety filter
&
Runtime blocking or modification of unsafe proposed actions
&
Shield or filter parameters
\\\\

\textbf{This work}
&
\textbf{Known feature map for an unknown stochastic linear MDP and an initial shield seed set}
&
\textbf{External PCIS-based shield updated from transition data}
&
\textbf{Conservative fixed point of a probabilistic safety operator (certifiable as an \((N,\epsilon)\)-PCIS) and a set-valued safe action map}
\\

\bottomrule
\end{tabularx}
\end{threeparttable}
\end{table*}

\textbf{Related work.}
Table~I compares the approaches most relevant to our setting along three axes:
what must be specified a-priori, how safety is enforced during learning,
and what safety object is synthesized from data.

On the RL side, recovery/reset-based methods guarantee safety during learning
by assuming an auxiliary fallback mechanism such as a trusted recovery controller,
a reset policy, or a learned recovery model \cite{Tommaso:17,thananjeyan2021recovery}.
CMDP/CPO-style approaches instead incorporate safety directly into the policy
update and typically target expected or cumulative constraint satisfaction for the
learned policy \cite{achiam2017cpo,chow2018lyapunov}. Safe-region expansion and
abstraction-based methods protect exploration by growing a certified safe region
or maintaining a safety abstraction under uncertainty \cite{hashimoto:23,HASHIMOTO2022110646}.
These lines of work are highly relevant, but they typically do not yield an external
runtime shield together with an explicit probabilistic invariant-set object
recomputed from transition data.

A particularly relevant control-theoretic line of work studies data-driven synthesis
or verification of invariant sets. \cite{korda2020ciset} computes approximations of maximum positively
invariant and maximum controlled invariant sets directly from one-step transitions via
convex optimization. For uncertain linear and LPV systems,
data-driven robust invariant or robust control invariant sets and associated controllers
can be synthesized directly from data, without first identifying a nominal model
\cite{chen2018mrci,mulagaleti2022ris,mejari2023rci}. For nonlinear systems, recent
works study positive invariant-set synthesis or verification directly from sampled data
\cite{kashani2025invariant,strong2024verification}. Closest in spirit to the
probabilistic object considered here is \cite{griffioen2024gp}, which computes probabilistic controlled invariant sets for nonlinear systems together with feedback controllers. In contrast, we
focus on unknown stochastic linear MDPs, where hard invariance guarantee is often too restrictive and
an \((N,\epsilon)\)-PCIS is the more natural object. The linear-MDP structure lets us estimate the probabilistic predecessor operator directly from transition data using least squares and confidence bounds, and thereby produce not only a conservative fixed-point safety set but also a set-valued safe action map that can serve as an external shield for RL.

Shielding and safety-filtering methods are also closely related because they intervene
online by blocking or modifying unsafe proposed actions
\cite{alshiekh2018safe,jansen2020probshield,konighofer2023onlineshield,wabersich2021psf,goodall2023approxshield}.
In most existing works, however, the shield or filter is either given a-priori or
constructed from a safety-relevant model fragment. Our contribution is complementary:
starting from an initial PCIS seed, we reconstruct the shield directly from accumulated
transition data via a conservative backward recursion. On the statistical-modeling side,
several recent safe RL works also exploit linear function approximation or linear-MDP-type
structure under unknown constraints or hard safety conditions
\cite{amani2021safe,ghosh2022modelfree,shi2023hard,Wei:2024}; their guarantees are
usually phrased in terms of regret, cumulative constraint violations, or safe action
selection inside the RL update, rather than in terms of an external shield built from
an explicit probabilistic invariant-set object.

In summary, our main contributions are as follows.
\begin{itemize}
  \item We derive an \(\eta\)-conservative approximation of the \(N\)-step safety (or, predecessor) operator for unknown linear MDPs using regularized least squares and self-normalized confidence bounds.
  \item We show that a conservative fixed point of this operator, when the corresponding conservative-inclusion event holds, is certified as an \((N,\epsilon)\)-PCIS.
  \item We propose a lattice abstraction for continuous state spaces and incorporate an explicit Lipschitz-based discretization error term, yielding a tractable approximation scheme in continuous-state settings.
\end{itemize}

\smallskip
\noindent
\textbf{Notation.} Let \(\mathbb{N}\), \(\mathbb{N}_{\geq 0}\), \(\mathbb{N}_{>0}\), and \(\mathbb{N}_{a:b}\) denote the sets of integers, non-negative integers, positive integers, and integers in the interval \([a,b]\), respectively.
Let \(\mathbb{R}\), \(\mathbb{R}_{\geq 0}\), \(\mathbb{R}_{>0}\), and \(\mathbb{R}_{a:b}\) denote the sets of real numbers, non-negative real numbers, positive real numbers, and real numbers in the interval \([a,b]\), respectively.
We denote by \(\|\cdot\|\) and \(\|\cdot\|_\infty\) the Euclidean norm and the \(\infty\)-norm, respectively.
Given \(\Omega \subset \mathbb{R}^n\), let \(\1_{\Omega}:\mathbb{R}^n \rightarrow \{0,1\}\) denote the indicator function defined by \(\1_{\Omega}(x)=1\) if \(x\in\Omega\) and \(\1_{\Omega}(x)=0\) otherwise.
Given \(\delta_x>0\) and \(x\in\mathbb{R}^n\), let \(\Ball{\delta_x}{x}\subset\mathbb{R}^n\) denote the \(\infty\)-norm ball centered at \(x\), i.e.,
\[
\Ball{\delta_x}{x}=\{x' \in \mathbb{R}^n : \|x-x'\|_\infty \le \delta_x\}.
\]
For a signed finite measure \(\mu\), \(\|\mu\|_{\mathrm{TV}}\) denotes its total variation norm.

\section{Problem Setup}\label{sec:problem_setup}
We consider a discrete-time Markov decision process (MDP)
\(
\mathcal{M}=(\mathcal{X},\mathcal{U},\mathbb{P},\mathcal{X}_S)
\),
where \(\mathcal{X}\subseteq\mathbb{R}^n\) is a (possibly infinite) Borel state space,
\(\mathcal{U}\) is a finite action set,
\(\mathbb{P}(\cdot\mid x,u)\) is an unknown transition kernel on
\((\mathcal{X},\mathbb{B}(\mathcal{X}))\),
and \(\mathcal{X}_S\subseteq\mathcal{X}\) is the safety set. Although the transition kernel is unknown, we assume that it admits a linear representation with respect to a known feature map \(\phi\) \cite{jin20a}.

\begin{assumption}[Linear MDP structure]\label{assumption:linear}
\normalfont
There exist a known feature map
\(
\phi:\mathcal{X}\times\mathcal{U}\rightarrow\mathbb{R}^d
\)
and unknown signed finite measures
\(
\nu_1,\ldots,\nu_d
\)
on \((\mathcal{X},\mathbb{B}(\mathcal{X}))\) such that, for every measurable set
\(
A\in\mathbb{B}(\mathcal{X})
\)
and every
\(
(x,u)\in\mathcal{X}\times\mathcal{U}
\),
\begin{equation}
\mathbb{P}(A\mid x,u)
=
\sum_{\ell=1}^{d}\phi_\ell(x,u)\,\nu_\ell(A)
=
\bigl\langle \phi(x,u),\nu(A)\bigr\rangle,
\label{eq:linear_mdp}
\end{equation}
where
\(
\nu(A)=[\nu_1(A),\ldots,\nu_d(A)]^\top
\).
Moreover, \(\|\nu_\ell\|_{\mathrm{TV}}\le 1\) for all \(\ell\in\{1,\ldots,d\}\), and
\(\|\phi(x,u)\|_2\le 1\) for all \((x,u)\in\mathcal{X}\times\mathcal{U}\).
The feature map is Lipschitz continuous in the state variable: there exists \(L_\phi\ge 0\) such that
\begin{equation}
\bigl|\phi_\ell(x,u)-\phi_\ell(x',u)\bigr|
\le
L_\phi\|x-x'\|_\infty
\label{eq:phi_lipschitz}
\end{equation}
for all \(x,x'\in\mathcal{X}\), \(u\in\mathcal{U}\), and \(\ell\in\{1,\ldots,d\}\).
\end{assumption}
In general, the motivation of employing the linear-MDP in RL is that it provides a finite-dimensional representation of the transition kernel that enables statistically efficient generalization across state-action pairs in large or continuous spaces. This assumption has become a standard tractable non-tabular model in RL: when the feature map is known, optimistic methods (e.g., upper confidence bound (UCB)-based approaches) achieve regret bounds polynomial in the feature dimension and horizon and \textit{independent} of the cardinality of the state and action spaces \cite{jin20a,zhou21a}. Recent work further shows that the required features can themselves be learned in a provably efficient and practically effective manner \cite{zhang22x}. We therefore adopt the linear-MDP model not as a claim that the physical plant is literally linear, but as a compromise between expressivity, sample efficiency, and computational tractability. 
Compared with black-box assumptions, this structure is strong enough to support finite-sample confidence sets and dynamic-programming-based planning, as we will see in the next section.

For a horizon \(N\in\mathbb{N}\), we consider finite-horizon nonstationary deterministic Markov policies
\[
\Pi_N=(\mu_0,\mu_1,\ldots,\mu_{N-1}),
\]
where each \(\mu_k:\mathcal{X}\to\mathcal{U}\).
Starting from \(x_0\in\mathcal{X}\), the resulting trajectory satisfies
\[
u_k=\mu_k(x_k),\ 
x_{k+1}\sim\mathbb{P}(\cdot\mid x_k,u_k),\ k=0,\ldots,N-1.
\]
For any \(\Omega\subseteq\mathcal{X}_S\), policy \(\Pi_N\), and initial state \(x_0\in\Omega\), define
\begin{equation}
p^\Omega_{\Pi_N}(x_0)
=
\Pr \left(
x_k\in\Omega,\ \forall k=0,\ldots,N
\right).
\label{eq:pdef}
\end{equation}
Thus, \(N\)-step safety means safety over \(N\) transitions, including the terminal state \(x_N\).

\begin{definition}[$(N,\epsilon)$-PCIS]\label{def:CIS}
\normalfont
A set \(\mathcal{X}_{\mathrm{PCIS}}\subseteq\mathcal{X}_S\) is an
\((N,\epsilon)\)-\emph{probabilistic controlled invariant set} (PCIS) if, for every
\(x\in\mathcal{X}_{\mathrm{PCIS}}\), there exists a deterministic Markov policy \(\Pi_N\) such that
\begin{equation}
p^{\mathcal{X}_{\mathrm{PCIS}}}_{\Pi_N}(x)\ge 1-\epsilon.
\label{eq:CIS}
\end{equation}
\end{definition}

For \(\Omega\subseteq\mathcal{X}_S\), let \(\mathcal{Q}^N_\epsilon(\Omega)\) denote the exact safety operator
\begin{equation}
\mathcal{Q}^N_\epsilon(\Omega)
=
\Bigl\{
x\in\Omega:
\exists \Pi_N\ \text{s.t.}\ p^\Omega_{\Pi_N}(x)\ge 1-\epsilon
\Bigr\}.
\label{eq:Q_exact}
\end{equation}
Starting from \(\Omega_0=\mathcal{X}_S\), define the decreasing sequence
\[
\Omega_{j+1}=\mathcal{Q}^N_\epsilon(\Omega_j),\ \  j=0,1,\ldots.
\]
Since \(\Omega_{j+1}\subseteq\Omega_j\), the maximal \((N,\epsilon)\)-PCIS in \(\mathcal{X}_S\) is given by
\[
\mathcal{X}^*_{\mathrm{PCIS}}
=
\bigcap_{j=0}^{\infty}\Omega_j,
\]
see, e.g., \cite{yulong2021}.
Thus, if the exact operator \textit{were} available, one could characterize a PCIS as a fixed point of
\(\mathcal{Q}^N_\epsilon\).

In the setting considered here, where the transition kernel is unknown, the exact operator $\mathcal{Q}^N_\epsilon(\Omega)$ is indeed not directly available and will thus be inferred from data. This motivates taking a fixed point of a \textit{conservative} approximation as the primary computable object and interpreting it as a \textit{candidate}
 $(N, \epsilon)$-{PCIS}.
\begin{definition}[$\eta$-conservative approximation of the operator/candidate PCIS]\label{def:conservative}
\normalfont
Fix \(\Omega\subseteq\mathcal{X}_S\).
A set
\[
\widetilde{\mathcal{Q}}^N_\epsilon(\Omega)\subseteq\Omega
\]
is called an \emph{\(\eta\)-conservative approximation} of
\(\mathcal{Q}^N_\epsilon(\Omega)\) if
\begin{equation}
\Pr\left(
\widetilde{\mathcal{Q}}^N_\epsilon(\Omega)
\subseteq
\mathcal{Q}^N_\epsilon(\Omega)
\right)
\ge \eta.
\label{eq:Q_conservative}
\end{equation}
Moreover, a set
\[
\widehat{\Omega}\subseteq\mathcal{X}_S
\]
is called a \emph{conservative fixed point} if
\begin{equation}
\widehat{\Omega}
=
\widetilde{\mathcal{Q}}^N_\epsilon(\widehat{\Omega}).
\label{eq:conservative_fixed_point}
\end{equation}
Such a set is referred to as a \textit{candidate} $(N, \epsilon)$-{PCIS}.
\end{definition}

Let us emphasize that the conservative fixed point $\widehat{\Omega}$ is, at this stage, only a
\textit{candidate} $(N,\epsilon)$-PCIS.
The guarantee in \eqref{eq:Q_conservative} is stated for a prescribed reference set $\Omega$
and therefore does not automatically apply to the fixed point $\widehat{\Omega}$ returned by the
fixed-point search.
A separate certification argument is thus needed, and will be provided in the next section.
\begin{problem}\label{problem:objective}
\normalfont
Under \cref{assumption:linear}, given a set
\(
\Omega\subseteq\mathcal{X}_S
\),
a horizon
\(
N\in\mathbb{N}
\),
a safety tolerance
\(
\epsilon\in(0,1)
\),
a confidence level
\(
\eta\in(0,1)
\),
and a transition dataset
\[
\mathcal{D}_T=\{(x_t,u_t,x_{t+1})\}_{t=0}^{T-1}\subseteq\mathcal{X}_S\times\mathcal{U}\times\mathcal{X}_S,
\]
construct an \(\eta\)-conservative approximation
\(
\widetilde{\mathcal{Q}}^N_\epsilon(\Omega)
\)
of the exact safety operator
\(
\mathcal{Q}^N_\epsilon(\Omega)
\), as well as the associated candidate $(N, \epsilon)$-PCIS.
\end{problem}


\section{Main Results}\label{sec:main_results}

\subsection{Computing $\eta$-conservative operator for linear MDPs}

Let us first provide a standard consequence of \cref{assumption:linear}.

\begin{lemma}\label{lemma:expectation}
\normalfont
Let \cref{assumption:linear} hold and let \(p:\mathcal{X}\to[0,1]\) be measurable.
Then there exists \(\theta_p\in\mathbb{R}^d\) such that, for all \((x,u)\in\mathcal{X}\times\mathcal{U}\),
\begin{equation}
\mathbb{E}[p(x')\mid x,u]
=
\langle \theta_p,\phi(x,u)\rangle,
\label{eq:lin_expect}
\end{equation}
where \(x'\sim\mathbb{P}(\cdot\mid x,u)\).
\end{lemma}
\begin{proof}
For each \(\ell\), define
\(
\theta_{p,\ell}=\int_{\mathcal{X}} p(x')\,\nu_\ell(dx')
\).
Since \(p\in[0,1]\) and \(\|\nu_\ell\|_{\mathrm{TV}}\le 1\), each integral is finite and
\(
|\theta_{p,\ell}|
\le
\|p\|_\infty \|\nu_\ell\|_{\mathrm{TV}}
\le 1
\).
Let
\(
\theta_p=[\theta_{p,1},\ldots,\theta_{p,d}]^\top
\).
Then, using \cref{assumption:linear},
\begin{align}
\mathbb{E}[p(x')\mid x,u]
&=
\int_{\mathcal{X}} p(x')\,\mathbb{P}(dx'\mid x,u) \notag \\ 
&=
\sum_{\ell=1}^{d}\phi_\ell(x,u)\int_{\mathcal{X}} p(x')\,\nu_\ell(dx')\notag \\ 
&=
\langle \theta_p,\phi(x,u)\rangle.
\end{align}
This proves \eqref{eq:lin_expect}.
\end{proof}

For a fixed \(\Omega\subseteq\mathcal{X}_S\), define the dynamic-programming recursion
\begin{align}
p_N^\Omega(x)
&=
\1_\Omega(x), \notag\\
p_j^\Omega(x)
&=
\1_\Omega(x)\,
\max_{u\in\mathcal{U}}
\mathbb{E}[p_{j+1}^\Omega(x')\mid x,u],
\ j=N-1,\ldots,0.
\label{eq:p_exact_rec}
\end{align}
Then,
\(
\sup_{\Pi_N} p^\Omega_{\Pi_N}(x)=p_0^\Omega(x)
\)
(see, e.g., \cite{yulong2021}), and therefore
\begin{equation}
\mathcal{Q}^N_\epsilon(\Omega)
=
\{x\in\Omega:\ p_0^\Omega(x)\ge 1-\epsilon\}.
\label{eq:Q_exact_dp}
\end{equation}
The index \(j\in\{0,\ldots,N-1\}\) refers to the
\emph{backward stage} of the recursion \eqref{eq:p_exact_rec}. In other words, stage \(j\) computes (or estimates) the conditional expectation of the stage-\((j+1)\) continuation value.
To avoid the dependence issue caused by reusing the same data across all backward stages, we use stagewise sample splitting in the theoretical guarantee.
Towards this end, let
\[
\mathcal{D}^{(j)}
=
\{(x_t^{(j)},u_t^{(j)},x_{t+1}^{(j)})\}_{t=0}^{T_j-1},
\ j=0,\ldots,N-1,
\]
be independent datasets. Here, these datasets are assumed to be obtained from independent data-collection procedures rather than by post hoc partitioning of a single replay buffer; practical procedures to enforce such independence are discussed later in \cref{rem:shield_independence}.

For each \(j\), define
\begin{equation}
D_j
=
\begin{bmatrix}
\phi(x_0^{(j)},u_0^{(j)})^\top\\
\vdots\\
\phi(x_{T_j-1}^{(j)},u_{T_j-1}^{(j)})^\top
\end{bmatrix},
\ 
V_j=D_j^\top D_j + I_d.
\label{eq:DjVj}
\end{equation}
and
\begin{equation}
\sigma_j(x,u)
=
\sqrt{\phi(x,u)^\top V_j^{-1}\phi(x,u)}.
\label{eq:sigmaj}
\end{equation}
The next result is stated as a reduction theorem: any choice of the parameter that makes the lower-confidence events below hold yields a conservative operator.
\begin{theorem}[$\eta$-conservative approximation operator guarantee]\label{theorem:Qtilde}
\normalfont
Let \cref{assumption:linear} hold and fix \(\Omega\subseteq\mathcal{X}_S\).
Let \(\delta_j\in(0,1)\), \(j=0,\ldots,N-1\), satisfy
\(
\sum_{j=0}^{N-1}\delta_j \le 1-\eta
\). Let \(\tilde p_N:\mathcal{X}\to[0,1]\) be given by
\begin{equation}
\tilde p_N(x)=\1_\Omega(x).
\label{eq:ptilde_terminal}
\end{equation}
For \(j=N-1,\ldots,0\), let
\begin{equation}
y_t^{(j)}
=
\tilde p_{j+1}(x_{t+1}^{(j)}),\ \ 
\hat\theta_j
=
V_j^{-1}D_j^\top y^{(j)},
\label{eq:ridge_est}
\end{equation}
where
\(
y^{(j)}=[y_0^{(j)},\ldots,y_{T_j-1}^{(j)}]^\top
\).
Assume that \(\beta_j=\beta_j(\delta_j)\) is chosen so that the lower-confidence event
\begin{equation}
\label{eq:Ej}
\mathcal{E}_j =
\Biggl\{
\begin{aligned}
&\forall (x,u)\in\mathcal{X}\times\mathcal{U}:\\
&\mathbb{E}[\tilde p_{j+1}(x')\mid x,u]\ge \ell_j(x,u)
\end{aligned}
\Biggr\},
\end{equation}
where $\ell_j(x,u) =
\hat\theta_j^\top\phi(x,u)-\beta_j\sigma_j(x,u)$, 
satisfies\footnote{Indeed, the lower-confidence event in \eqref{eq:Ej} is obtained by applying a standard self-normalized concentration bound for regularized least squares with conditionally sub-Gaussian noise; see, e.g.,\cite{abbasi2011online}.}
\begin{equation}
\Pr(\mathcal{E}_j\mid \tilde p_{j+1})\ge 1-\delta_j.
\label{eq:Ej_prob}
\end{equation}
Finally, for \(j=N-1,\ldots,0\), define
\begin{equation}
\tilde p_j(x)
=
\1_\Omega(x)\,
\max_{u\in\mathcal{U}}
\sat{\ell_j(x,u)}.
\label{eq:ptilde_rec}
\end{equation}
Then, the induced operator
\begin{equation}
\widetilde{\mathcal{Q}}^N_\epsilon(\Omega)
=
\{x\in\Omega:\ \tilde p_0(x)\ge 1-\epsilon\}
\label{eq:Qtilde_set}
\end{equation}
is an \(\eta\)-conservative approximation of \(\mathcal{Q}^N_\epsilon(\Omega)\), i.e.,
\begin{equation}
\Pr \left(
\widetilde{\mathcal{Q}}^N_\epsilon(\Omega)
\subseteq
\mathcal{Q}^N_\epsilon(\Omega)
\right)
\ge \eta.
\label{eq:Qtilde_conservative}
\end{equation}
\end{theorem}

\begin{proof}
For each stage \(j\), since \(\tilde p_{j+1}\) is computed from later stages only, it is independent of the dataset \(\mathcal{D}^{(j)}\).
Conditioned on \(\tilde p_{j+1}\), \cref{lemma:expectation} implies that there exists
\(
\theta_j\in\mathbb{R}^d
\)
such that
\[
\mathbb{E}[\tilde p_{j+1}(x')\mid x,u]
=
\theta_j^\top \phi(x,u)
\]
for all \((x,u)\). By backward induction from \eqref{eq:ptilde_terminal} and \eqref{eq:ptilde_rec},
\begin{equation}
0 \le \tilde p_j(x) \le 1,\ \forall x\in\mathcal{X},\ \forall j=0,\ldots,N.
\label{eq:ptilde_bounded}
\end{equation}
Hence, the stage-\(j\) regression targets satisfy
\(
y_t^{(j)}=\tilde p_{j+1}(x_{t+1}^{(j)})\in[0,1]
\)
for all \(t\).
Let
\(
\mathcal{E}=\bigcap_{j=0}^{N-1}\mathcal{E}_j
\).
By \eqref{eq:Ej_prob} and the union bound,
\[
\Pr(\mathcal{E})
\ge
1-\sum_{j=0}^{N-1}\delta_j
\ge \eta.
\]
We now prove by backward induction that, on \(\mathcal{E}\),
\begin{equation}
\tilde p_j(x)\le p_j^\Omega(x),\ \ 
\forall x\in\mathcal{X},\ \forall j=0,\ldots,N.
\label{eq:induction_goal}
\end{equation}
The claim is immediate at \(j=N\) since \(\tilde p_N=p_N^\Omega=\1_\Omega\).
Assume it holds at stage \(j+1\).
Then, for any \((x,u)\),
\[
\mathbb{E}[p_{j+1}^\Omega(x')\mid x,u]
\ge
\mathbb{E}[\tilde p_{j+1}(x')\mid x,u].
\]
On \(\mathcal{E}_j\),
\[
\mathbb{E}[\tilde p_{j+1}(x')\mid x,u]
\ge
\hat\theta_j^\top\phi(x,u)-\beta_j\sigma_j(x,u).
\]
Therefore,
\[
\mathbb{E}[p_{j+1}^\Omega(x')\mid x,u]
\ge
\hat\theta_j^\top\phi(x,u)-\beta_j\sigma_j(x,u).
\]
Taking the maximum over \(u\), multiplying by \(\1_\Omega(x)\), and clipping to \([0,1]\) yields
\(
p_j^\Omega(x)\ge \tilde p_j(x)
\),
which proves \eqref{eq:induction_goal}.
Therefore, on \(\mathcal{E}\),
\[
\tilde p_0(x)\le p_0^\Omega(x),\ \ \forall x\in\mathcal{X}.
\]
Consequently,
\begin{align}
\widetilde{\mathcal{Q}}^N_\epsilon(\Omega)
&=
\{x\in\Omega:\tilde p_0(x)\ge 1-\epsilon\} \notag \\
&\subseteq
\{x\in\Omega:p_0^\Omega(x)\ge 1-\epsilon\}
=
\mathcal{Q}^N_\epsilon(\Omega).
\end{align}
Since \(\Pr(\mathcal{E})\ge \eta\), \eqref{eq:Qtilde_conservative} follows.
\end{proof}

The discussion above guarantees conservative inclusion only for a prescribed fixed reference set \(\Omega\). In other words, it does not by itself imply that a data-driven conservative fixed point, selected using the same data, is an \((N,\epsilon)\)-PCIS.
A natural way to obtain a rigorous high-confidence guarantee is to separate the data used to construct the set from the data used to certify it. 
\begin{corollary}[Certification of PCIS by sample splitting]
\label{theorem:independent_cert}
\normalfont
Let
\(
\mathcal{D}^{\mathrm{grow}}
\)
and
\(
\mathcal{D}^{\mathrm{cert}}
\)
be independent transition datasets.
Let
\(
\widehat{\Omega}_g
\subseteq \mathcal{X}_S
\)
be any $(N,\epsilon)$-PCIS candidate set constructed using only
\(
\mathcal{D}^{\mathrm{grow}}
\).
Assume that, using only
\(
\mathcal{D}^{\mathrm{cert}}
\),
construct an \(\eta\)-conservative approximation
\(
\widetilde{\mathcal{Q}}^{N,\mathrm{cert}}_\epsilon
\)
of
\(
\mathcal{Q}^N_\epsilon
\), i.e., for each $\Omega \subseteq \mathcal{X}_S$, 
\begin{equation}
\Pr\!\left(
\widetilde{\mathcal{Q}}^{N,\mathrm{cert}}_\epsilon(\Omega)
\subseteq
\mathcal{Q}^N_\epsilon(\Omega)
\right)
\ge \eta .
\label{eq:independent_cert_conservative}
\end{equation}
Then,
\begin{equation}
\Pr\!\left(
\widehat{\Omega}_g
\subseteq
\widetilde{\mathcal{Q}}^{N,\mathrm{cert}}_\epsilon(\widehat{\Omega}_g)
\ \Longrightarrow\
\widehat{\Omega}_g\ \text{is an }(N,\epsilon)\text{-PCIS}
\right)
\ge \eta.
\label{eq:independent_cert_unconditional}
\end{equation}
\end{corollary}
\begin{proof}
For a fixed
\(
\mathcal{D}^{\mathrm{grow}}
\), denote the resulting $(N,\epsilon)$-PCIS candidate set by
\(
\widehat{\Omega}_g
\subseteq \mathcal{X}_S
\). Notice that,
conditioned on
\(
\mathcal{D}^{\mathrm{grow}}
\),
the set
\(
\widehat{\Omega}_g
\)
is deterministic.
Since
\(
\mathcal{D}^{\mathrm{cert}}
\)
is independent of
\(
\mathcal{D}^{\mathrm{grow}}
\),
the conservative-inclusion guarantee \eqref{eq:independent_cert_conservative} applies to the fixed set
\(
\widehat{\Omega}_g
\),
and hence
\[
\Pr\!\left(
\widetilde{\mathcal{Q}}^{N,\mathrm{cert}}_\epsilon(\widehat{\Omega}_g)
\subseteq
\mathcal{Q}^N_\epsilon(\widehat{\Omega}_g)
\;\middle|\;
\mathcal{D}^{\mathrm{grow}}
\right)
\ge \eta
\]
almost surely. Now, consider the event
\[
\widehat{\Omega}_g
\subseteq
\widetilde{\mathcal{Q}}^{N,\mathrm{cert}}_\epsilon(\widehat{\Omega}_g).
\]
On the intersection of this event with
\[
\widetilde{\mathcal{Q}}^{N,\mathrm{cert}}_\epsilon(\widehat{\Omega}_g)
\subseteq
\mathcal{Q}^N_\epsilon(\widehat{\Omega}_g),
\]
we obtain
\[
\widehat{\Omega}_g
\subseteq
\widetilde{\mathcal{Q}}^{N,\mathrm{cert}}_\epsilon(\widehat{\Omega}_g)
\subseteq
\mathcal{Q}^N_\epsilon(\widehat{\Omega}_g).
\]
On the other hand, by definition of the exact operator, $\mathcal{Q}^N_\epsilon(\widehat{\Omega}_g)\subseteq \widehat{\Omega}_g$. Thus, 
\[
\widehat{\Omega}_g
\subseteq
\mathcal{Q}^N_\epsilon(\widehat{\Omega}_g)
\subseteq
\widehat{\Omega}_g,
\]
which implies
\[
\widehat{\Omega}_g
=
\mathcal{Q}^N_\epsilon(\widehat{\Omega}_g).
\]
By the fixed-point characterization of \((N,\epsilon)\)-PCIS, this shows that
\(
\widehat{\Omega}_g
\)
is an \((N,\epsilon)\)-PCIS.
Hence, conditioned on
\(
\mathcal{D}^{\mathrm{grow}}
\),
the implication
\[
\widehat{\Omega}_g
\subseteq
\widetilde{\mathcal{Q}}^{N,\mathrm{cert}}_\epsilon(\widehat{\Omega}_g)
\ \Longrightarrow\
\widehat{\Omega}_g\ \text{is an }(N,\epsilon)\text{-PCIS}
\]
holds with probability at least \(\eta\) over
\(
\mathcal{D}^{\mathrm{cert}}
\).
\end{proof}

\subsection{A tractable lattice abstraction for operator evaluation}

When \(\mathcal{X}\) is infinite, evaluating \eqref{eq:Qtilde_set} over all \(x\in\mathcal{X}_S\) is intractable.
We therefore discretize the reference set.
For \(\delta_x>0\), define the lattice
\[
[\mathcal{X}_S]_{\delta_x}
=
\{x\in\mathcal{X}_S:\ x_i=a_i\delta_x,\ a_i\in\mathbb{Z},\ i=1,\ldots,n\}.
\]
Let
\(
q_{\delta_x}:\mathcal{X}_S\to [\mathcal{X}_S]_{\delta_x}
\)
be a measurable nearest-neighbor quantizer in \(\|\cdot\|_\infty\), i.e.,
\begin{equation}
q_{\delta_x}(x)\in
\argmin_{x_d\in[\mathcal{X}_S]_{\delta_x}}
\|x-x_d\|_\infty.
\label{eq:quantizer_nn}
\end{equation}
We assume that \([\mathcal{X}_S]_{\delta_x}\) forms a \(\delta_x\)-net of \(\mathcal{X}_S\), so that
\begin{equation}
\|x-q_{\delta_x}(x)\|_\infty\le \delta_x,\ \forall x\in\mathcal{X}_S.
\label{eq:delta_net}
\end{equation}
The following result shows that the lattice-based approximation remains conservative after accounting for the discretization error.
\begin{theorem}[Lattice-based conservative operator]\label{theorem:abstraction}
\normalfont
Let \cref{assumption:linear} hold and fix \(\Omega\subseteq\mathcal{X}_S\).
Use the independent datasets \(\mathcal{D}^{(j)}\) and the matrices
\(D_j,V_j,\sigma_j\) from \eqref{eq:DjVj}--\eqref{eq:sigmaj}.
Let \(\delta_j\in(0,1)\), \(j=0,\ldots,N-1\), satisfy
\(
\sum_{j=0}^{N-1}\delta_j \le 1-\eta
\).
For \(j=N\), define
\begin{equation}
\tilde p_N^{\delta_x}(x_d)=\1_\Omega(x_d),\ 
x_d\in[\mathcal{X}_S]_{\delta_x},
\label{eq:ptilde_dx_terminal}
\end{equation}
and its lift
\begin{equation}
\bar p_N^{\delta_x}(x)
=
\begin{cases}
\tilde p_N^{\delta_x}(q_{\delta_x}(x)), & x\in\Omega,\\
0, & x\notin\Omega.
\end{cases}
\label{eq:pbar_dx_terminal}
\end{equation}
For \(j=N-1,\ldots,0\), define the stage-\(j\) regression targets
\begin{align}
&y_t^{(j),\delta_x}
=
\bar p_{j+1}^{\delta_x}(x_{t+1}^{(j)}), \notag\\
&\hat\theta_j^{\delta_x}
=
V_j^{-1}D_j^\top y^{(j),\delta_x},
\label{eq:ridge_est_dx}
\end{align}
where
\(
y^{(j),\delta_x}
=
[y_0^{(j),\delta_x},\ldots,y_{T_j-1}^{(j),\delta_x}]^\top
\).
Assume that \(\beta_j=\beta_j(\delta_j)\) is chosen so that the event
\begin{equation}
\label{eq:Ej_dx}
\mathcal{E}_j^{\delta_x}
=
\Biggl\{
\begin{aligned}
&\forall (x,u)\in\mathcal{X}\times\mathcal{U}:\\
&\mathbb{E}[\bar p_{j+1}^{\delta_x}(x')\mid x,u] \ge \ell^{\delta_x}_j(x,u)
\end{aligned}
\Biggr\},
\end{equation}
where $\ell_j^{\delta_x}(x_d,u)
=
(\hat\theta_j^{\delta_x})^\top\phi(x_d,u)
-dL_\phi\delta_x
-\beta_j\sigma_j(x_d,u)$, satisfies
\begin{equation}
\Pr(\mathcal{E}_j^{\delta_x}\mid \bar p_{j+1}^{\delta_x})\ge 1-\delta_j.
\label{eq:Ej_dx_prob}
\end{equation}
Then define, for \(x_d\in[\mathcal{X}_S]_{\delta_x}\),
\begin{equation}
\tilde p_j^{\delta_x}(x_d)
=
\1_\Omega(x_d)\,
\max_{u\in\mathcal{U}}
\sat{\ell_j^{\delta_x}(x_d,u)},
\label{eq:ptilde_dx_rec}
\end{equation}
and its lift
\begin{equation}
\bar p_j^{\delta_x}(x)
=
\begin{cases}
\tilde p_j^{\delta_x}(q_{\delta_x}(x)), & x\in\Omega,\\
0, & x\notin\Omega.
\end{cases}
\label{eq:pbar_dx_rec}
\end{equation}
Finally, define
\begin{equation}
\widetilde{\mathcal{Q}}^N_{\delta_x,\epsilon}(\Omega)
=
\{x\in\Omega:\bar p_0^{\delta_x}(x)\ge 1-\epsilon\}.
\label{eq:Qtilde_dx}
\end{equation}
Then
\begin{equation}
\Pr \left(
\widetilde{\mathcal{Q}}^N_{\delta_x,\epsilon}(\Omega)
\subseteq
\mathcal{Q}^N_\epsilon(\Omega)
\right)
\ge \eta.
\label{eq:Qtilde_dx_conservative}
\end{equation}
\end{theorem}
For the proof, see Appendix. 

\subsection{Set-valued safe action maps}

For a fixed \(\Omega\), the conservative backward recursion can naturally induce a stage-dependent family of set-valued safe action maps.
When \(N>1\), this object should be viewed as a set-valued map on the augmented state \((j,x)\), rather than as a single stationary policy.

For a given \(\Omega\subseteq\mathcal{X}_S\), let \(\tilde p_j\) be computed as in \cref{theorem:Qtilde}.
For each stage \(j=0,\dots,N-1\) and \(x\in\Omega\), define
\begin{equation}
\mu_j^S(x;\Omega)
=
\bigl\{u\in\mathcal{U}:\ell_j(x,u)\ge 1-\epsilon\bigr\}.
\label{eq:setvalued_policy}
\end{equation}
Equivalently, we may regard the set-valued safe action map as the augmented-state map
\[
\mu^S(j,x;\Omega)=\mu_j^S(x;\Omega),\ (j,x)\in \{0,\dots,N-1\}\times \Omega.
\]

Similarly, in the lattice-based case, for each stage \(j=0,\dots,N-1\), \(x\in\Omega\), and \(x_d=q_{\delta_x}(x)\), define
\begin{equation}
\mu_j^{S,\delta_x}(x;\Omega)
=
\bigl\{u\in\mathcal{U}:\ell_j^{\delta_x}(x_d,u)\ge 1-\epsilon\bigr\}.
\label{eq:setvalued_policy_dx}
\end{equation}

\begin{corollary}[Safe actions for a fixed reference set]\label{cor:setvalued}
\normalfont
On the confidence event of \cref{theorem:Qtilde}, every action \(u\in\mu_j^S(x;\Omega)\) satisfies
\begin{equation}
\mathbb{E}[p_{j+1}^\Omega(x')\mid x,u]\ge 1-\epsilon.
\label{eq:safe_action_cert}
\end{equation}
Define the stage-\(j\) conservative predecessor set by
\[
\widetilde{\mathcal{Q}}_{\epsilon}^{N,j}(\Omega)
:=
\{x\in\Omega:\tilde p_j(x)\ge 1-\epsilon\},
\ \ j=0,\dots,N-1.
\]
Then, if \(x\in\widetilde{\mathcal{Q}}_{\epsilon}^{N,j}(\Omega)\), we have \(\mu_j^S(x;\Omega)\neq\varnothing\).
In particular,
\[
\widetilde{\mathcal{Q}}_{\epsilon}^{N,0}(\Omega)
=
\widetilde{\mathcal{Q}}_\epsilon^N(\Omega),
\]
so \(x\in\widetilde{\mathcal{Q}}_\epsilon^N(\Omega)\) implies \(\mu_0^S(x;\Omega)\neq\varnothing\).
The same statement holds for the lattice-based action sets in \eqref{eq:setvalued_policy_dx} under the event of \cref{theorem:abstraction}.
\end{corollary}

\begin{proof}
If \(u\in\mu_j^S(x;\Omega)\), then by definition
\[
\hat\theta_j^\top\phi(x,u)-\beta_j\sigma_j(x,u)\ge 1-\epsilon.
\]
On the confidence event of \cref{theorem:Qtilde},
\[
\mathbb{E}[\tilde p_{j+1}(x')\mid x,u]
\ge
\hat\theta_j^\top\phi(x,u)-\beta_j\sigma_j(x,u)
\ge 1-\epsilon.
\]
Since \(\tilde p_{j+1}\le p_{j+1}^\Omega\) on the same event, \eqref{eq:safe_action_cert} follows.
Now, let \(x\in\widetilde{\mathcal{Q}}_{\epsilon}^{N,j}(\Omega)\), so that \(\tilde p_j(x)\ge 1-\epsilon\).
By definition of \(\tilde p_j\),
\[
\tilde p_j(x)
=
\max_{u\in\mathcal U}\sat{\ell_j(x,u)},
\]
hence at least one action attains a value no smaller than \(1-\epsilon\), and therefore \(\mu_j^S(x;\Omega)\neq\varnothing\). The lattice-based case is identical, with \(\tilde p_{j+1}\) replaced by \(\bar p_{j+1}^{\delta_x}\) and \cref{theorem:abstraction} used in place of \cref{theorem:Qtilde}.
\end{proof}

\begin{remark}[Interpretation for \(N>1\)]\label{rem:setvalued}
\normalfont
For \(N=1\), the set-valued map \(\mu_0^S(\cdot;\Omega)\) can be used directly as a one-step safe shield.
For \(N>1\), the family \(\{\mu_j^S(\cdot;\Omega)\}_{j=0}^{N-1}\) is stage-dependent and should be interpreted on the augmented state \((j,x)\).
Membership \(u\in\mu_j^S(x;\Omega)\) certifies that taking action \(u\) at stage \(j\) is safe relative to the continuation value \(p_{j+1}^\Omega\), i.e., relative to the existence of a compatible continuation over the remaining \(N-j-1\) stages.
Accordingly, for general \(N>1\), the family \(\{\mu_j^S\}_{j=0}^{N-1}\) should be used together with a stage index and a compatible continuation selector extracted from the backward recursion, for example
\[
u_j^\star(x)\in\argmax_{u\in\mathcal U}\sat{\ell_j(x,u)}.
\]
The present corollary does not by itself imply that arbitrary independent selectors from the thresholded sets \(\mu_j^S(\cdot;\Omega)\) can be freely composed across stages while preserving the same \(N\)-step guarantee.
\end{remark}
\section{Application to Safe Reinforcement Learning via Shielding}
\label{sec:shielding}

We now use the safe action maps induced by conservative backward recursions as a runtime \emph{shield}
around an arbitrary RL learner (e.g., \cite{alshiekh2018safe}).
The overview of the shielding mechanism is illustrated in Fig.~\ref{fig:overview}.
The shield is external to the RL update rule: it filters the action generated by the RL learner,
while the underlying learning algorithm can be chosen freely.
Let \(\mu_{\mathrm{RL}}\) denote the \emph{proposal} policy produced by the RL learner.

Because the executed action is filtered by the shield, the relevant performance objective is the
discounted return of the resulting shielded closed loop:
\begin{equation}
J_{\mathrm{sh}}(\mu_{\mathrm{RL}})
=
\mathbb{E}\left[\sum_{t=0}^{\infty}\gamma^t r_t\right],
\label{eq:RL_obj}
\end{equation}
where \(\gamma\in(0,1)\) and the expectation is taken under the shielded execution policy and the environment dynamics.

For notational simplicity, the online shielding algorithms below are written for the case \(N=1\).
For general \(N>1\), the stage-dependent safe action maps
\(
\{\mu_j^S(\cdot;\Omega)\}_{j=0}^{N-1}
\)
are applied in a receding-horizon manner, exactly as discussed in \cref{rem:setvalued}.

\begin{assumption}[Initial shield set]
\label{assumption:init_pcis}
\normalfont
There exists a nonempty initial shield set
\(
\widehat{\Omega}_0\subseteq\mathcal{X}_S
\)
together with a corresponding initial set-valued safe action map
\(
\mu_0^S(\cdot;\widehat{\Omega}_0)
\)
such that, whenever \(x_t\in\widehat{\Omega}_0\) and \(u_t\in\mu_0^S(x_t;\widehat{\Omega}_0)\),
the next state satisfies \(x_{t+1}\in \widehat{\Omega}_0\) with probability at least \(1-\epsilon\).
\end{assumption}

\begin{figure}[t]
\centering
\includegraphics[width=0.99\linewidth]{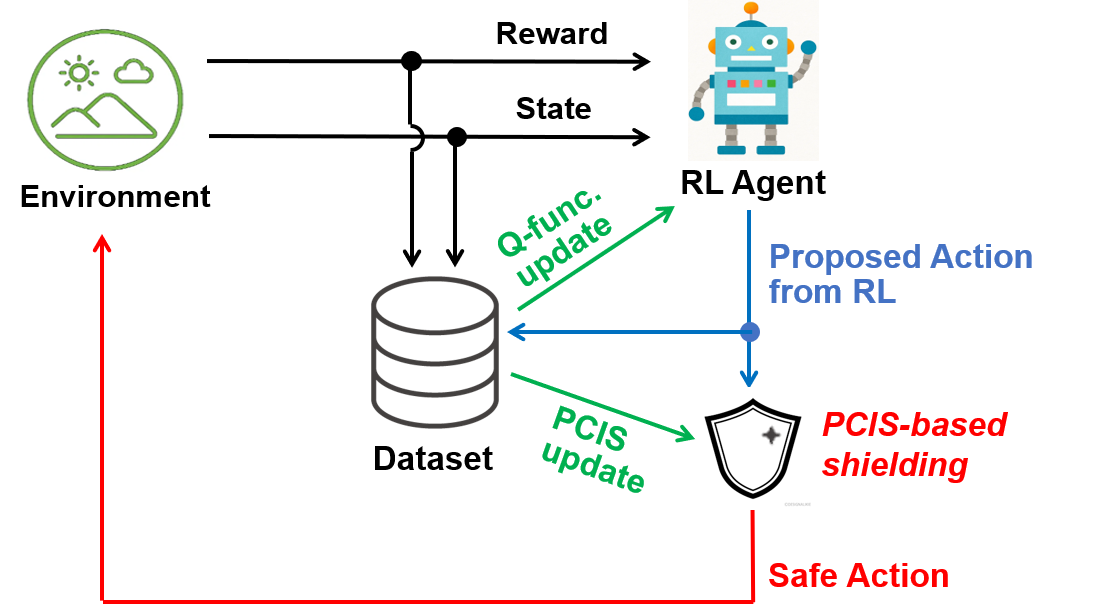}
\caption{Overview of safe RL via PCIS-based shielding with a grow/certify split.}
\label{fig:overview}
\end{figure}

At shielding update \(i\), let
\[
S_i(x)=\mu_i^S(x),\ \ x\in\widehat{\Omega}_i,
\]
where \((\widehat{\Omega}_i,\mu_i^S)\) denotes the \emph{current accepted shield}.
Given a proposal action \(u_t^{\mathrm{RL}}\), the executed action is defined by
\begin{equation}
u_t=
\begin{cases}
u_t^{\mathrm{RL}}, & \text{if } u_t^{\mathrm{RL}}\in S_i(x_t),\\[0.5ex]
\mathsf{backup}(x_t,S_i(x_t)), & \text{otherwise},
\end{cases}
\label{eq:shield_filter}
\end{equation}
where \(\mathsf{backup}(x,S)\in S\) is any selector.
For discrete action spaces, a natural choice is
\[
\mathsf{backup}(x,S)\in\argmax_{u\in S}Q_\theta(x,u),
\]
where \(Q_\theta\) is the current action-value estimate of the RL learner.
If \(\mathcal{U}\subseteq\mathbb{R}^m\), one may alternatively use a nearest-safe-action projection.
In continuous-state environments, both membership in \(\widehat{\Omega}_i\) and the safe action set \(S_i(x)\)
are evaluated via the lattice representation and the quantizer \(q_{\delta_x}\).
The filter \eqref{eq:shield_filter} is applied only when \(x_t\in\widehat{\Omega}_i\);
if \(x_t\notin\widehat{\Omega}_i\), the rollout is terminated or reset according to the task specification. To align the online shielding procedure with the sample-splitting certification result of the previous section,
we distinguish between two types of data at each shield update \(i\):

\begin{itemize}
\item a \emph{grow dataset} \(\mathcal{D}^{\mathrm{grow}}\), used to construct a tentative conservative fixed point, and
\item a fresh \emph{certification dataset} \(\mathcal{D}_i^{\mathrm{cert}}\), used only to certify that tentative set.
\end{itemize}

More precisely, the grow phase produces a tentative set
\(
\Omega_i^{\mathrm{tent}}
\subseteq \mathcal{X}_S
\)
by running the fixed-point search \(\textsc{ConInv}\) on \(\mathcal{D}^{\mathrm{grow}}\).
This tentative set is \emph{not} deployed immediately.
Instead, a separate certification step evaluates the conservative operator only once on the \emph{fixed}
reference set \(\Omega_i^{\mathrm{tent}}\) using \(\mathcal{D}_i^{\mathrm{cert}}\).
If the certification inclusion test succeeds, then \(\Omega_i^{\mathrm{tent}}\) is accepted as the next shield set.
In that case, the safe action map actually deployed online is the one extracted from the certification recursion.
Algorithm~\ref{alg:generic_shield} summarizes this learner-agnostic wrapper.

\renewcommand{\algorithmicrequire}{\textbf{Input:}}
\renewcommand{\algorithmicensure}{\textbf{Output:}}
\begin{algorithm}[t]
\caption{Safe RL via PCIS-based shielding with grow/certify splitting (generic template)}
\label{alg:generic_shield}
\begin{algorithmic}[1]
{\small
\REQUIRE{Initial shield \((\widehat{\Omega}_0,\mu_0^S)\), grow rollout length \(T_{\rm grow}\), certification rollout length \(T_{\rm cert}\), safe set \(\mathcal{X}_S\), certification data-collection protocol \(\beta_{\rm cert}\), backward-recursion parameters, and RL hyperparameters.}
\ENSURE{Learned RL parameters and current accepted shield \((\widehat{\Omega},\mu^S)\).}

\STATE Initialize RL parameters and accumulated grow dataset \(\mathcal{D}^{\mathrm{grow}}\leftarrow\varnothing\)
\STATE Set \(\widehat{\Omega}_i\leftarrow \widehat{\Omega}_0\), \(\mu_i^S\leftarrow \mu_0^S\), \(i\leftarrow 0\)

\REPEAT
  \STATE Start from the current environment state (or reset if the task is episodic)

  \FOR{$t=0$ \TO $T_{\rm grow}-1$}
    \STATE Propose \(u_t^{\mathrm{RL}}\) using the current RL learner
    \STATE Execute the shielded action \(u_t\) via \eqref{eq:shield_filter}
    \STATE Observe \(x_{t+1}\) and reward \(r_t\)
    \STATE Append \((x_t,u_t,x_{t+1})\) to \(\mathcal{D}^{\mathrm{grow}}\)
    \STATE Update the RL parameters using the chosen RL algorithm
    \STATE \(x_t\leftarrow x_{t+1}\)
  \ENDFOR

  \STATE \((\Omega_i^{\mathrm{tent}},\mu_i^{S,\mathrm{grow}})\leftarrow \textsc{ConInv}(\mathcal{D}^{\mathrm{grow}},\mathcal{X}_S)\)

  \STATE Collect a fresh certification dataset \(\mathcal{D}_i^{\mathrm{cert}}\) of length \(T_{\rm cert}\) using the certification protocol \(S_{\rm cert}\) from fresh resets
  \STATE Do \emph{not} use \(\mathcal{D}_i^{\mathrm{cert}}\) for RL updates or for the grow-phase fixed-point search

  \STATE \((a_i,\mu_i^{S,\mathrm{cert}})\leftarrow \textsc{CertifyShield}(\mathcal{D}_i^{\mathrm{cert}},\Omega_i^{\mathrm{tent}})\)

  \IF{$a_i=1$ \textbf{and} \(\widehat{\Omega}_i\subseteq \Omega_i^{\mathrm{tent}}\)}
    \STATE \(\widehat{\Omega}_{i+1}\leftarrow \Omega_i^{\mathrm{tent}}\)
    \STATE \(\mu_{i+1}^S\leftarrow \mu_i^{S,\mathrm{cert}}\)
  \ELSE
    \STATE \(\widehat{\Omega}_{i+1}\leftarrow \widehat{\Omega}_i\)
    \STATE \(\mu_{i+1}^S\leftarrow \mu_i^S\)
  \ENDIF

  \STATE \(i\leftarrow i+1\)
\UNTIL{stopping criterion is met}
}
\end{algorithmic}
\end{algorithm}

Algorithm~\ref{alg:shielded_dqn} instantiates Algorithm~\ref{alg:generic_shield} with DQN.
The proposal action is chosen by \(\varepsilon\)-greedy exploration with respect to the online network \(Q_\theta\),
while the replay buffer stores the \emph{executed} shielded transitions.
Hence, the Q-network is trained on the behavior actually induced by the current accepted shield.
The certification data are kept separate from the replay buffer and are never used for Q-learning updates.

\renewcommand{\algorithmicrequire}{\textbf{Input:}}
\renewcommand{\algorithmicensure}{\textbf{Output:}}
\begin{algorithm}[t]
\caption{Safe RL via PCIS-based shielding with grow/certify splitting (DQN-style example)}
\label{alg:shielded_dqn}
\begin{algorithmic}[1]
{\small
\REQUIRE{Initial shield \((\widehat{\Omega}_0,\mu_0^S)\), grow rollout length \(T_{\rm grow}\), certification rollout length \(T_{\rm cert}\), certification data-collection policy \(\beta_{\rm cert}\), safe set \(\mathcal{X}_S\), replay-buffer capacity \(M\), minibatch size \(B\), discount \(\gamma\), target-network update interval \(K\), and exploration schedule \(\varepsilon(\cdot)\).}
\ENSURE{Learned proposal policy \(\mu_{\rm RL}\) and current accepted shield \((\widehat{\Omega},\mu^S)\).}

\STATE Initialize online network \(Q_\theta\), target network \(Q_{\bar\theta}\leftarrow Q_\theta\)
\STATE Initialize replay buffer \(\mathcal{B}\leftarrow\varnothing\)
\STATE Initialize accumulated grow dataset \(\mathcal{D}^{\mathrm{grow}}\leftarrow\varnothing\)
\STATE Set \(\widehat{\Omega}_i\leftarrow \widehat{\Omega}_0\), \(\mu_i^S\leftarrow \mu_0^S\), \(i\leftarrow 0\)

\REPEAT
  \FOR{$t=0$ \TO $T_{\rm grow}-1$}
    \STATE Choose \(u_t^{\mathrm{RL}}\) from \(Q_\theta\) using \(\varepsilon(t)\)-greedy exploration
    \STATE Execute the shielded action \(u_t\) via \eqref{eq:shield_filter}
    \STATE Observe \(x_{t+1}\), reward \(r_t\), and terminal flag \(d_t\in\{0,1\}\)
    \STATE Store \((x_t,u_t,r_t,x_{t+1},d_t)\) in \(\mathcal{B}\)
    \STATE Append \((x_t,u_t,x_{t+1})\) to \(\mathcal{D}^{\mathrm{grow}}\)
    \IF{$|\mathcal{B}|\ge B$}
      \STATE Sample a minibatch from \(\mathcal{B}\)
      \STATE Perform one DQN update step
      \STATE Every \(K\) steps, update the target network \(Q_{\bar\theta}\)
    \ENDIF
    \STATE \(x_t\leftarrow x_{t+1}\)
  \ENDFOR

  \STATE \((\Omega_i^{\mathrm{tent}},\mu_i^{S,\mathrm{grow}})\leftarrow \textsc{ConInv}(\mathcal{D}^{\mathrm{grow}},\mathcal{X}_S)\)

  \STATE Collect a fresh certification dataset \(\mathcal{D}_i^{\mathrm{cert}}\) of length \(T_{\rm cert}\) using the fixed certification policy \(S_{\rm cert}\) from fresh resets
  \STATE Do \emph{not} insert \(\mathcal{D}_i^{\mathrm{cert}}\) into \(\mathcal{B}\), and do \emph{not} update \(Q_\theta\) with these samples

  \STATE \((a_i,\mu_i^{S,\mathrm{cert}})\leftarrow \textsc{CertifyShield}(\mathcal{D}_i^{\mathrm{cert}},\Omega_i^{\mathrm{tent}})\)

  \IF{$a_i=1$ \textbf{and} \(\widehat{\Omega}_i\subseteq \Omega_i^{\mathrm{tent}}\)}
    \STATE \(\widehat{\Omega}_{i+1}\leftarrow \Omega_i^{\mathrm{tent}}\)
    \STATE \(\mu_{i+1}^S\leftarrow \mu_i^{S,\mathrm{cert}}\)
  \ELSE
    \STATE \(\widehat{\Omega}_{i+1}\leftarrow \widehat{\Omega}_i\)
    \STATE \(\mu_{i+1}^S\leftarrow \mu_i^S\)
  \ENDIF

  \STATE \(i\leftarrow i+1\)
\UNTIL{stopping criterion is met}
}
\end{algorithmic}
\end{algorithm}

Algorithm~\ref{alg:coninv} details the grow-phase subroutine \(\textsc{ConInv}\).
Starting from \(\mathcal{X}_S\), it repeatedly applies the conservative operator approximation until a fixed point is reached.
The resulting fixed point is used only as a \emph{tentative} shield set.
In the continuous-state case, the operator is evaluated on the lattice abstraction from \cref{theorem:abstraction},
and the grow-phase safe action map returned at termination is extracted from the final backward recursion.

\renewcommand{\algorithmicrequire}{\textbf{Input:}}
\renewcommand{\algorithmicensure}{\textbf{Output:}}
\begin{algorithm}[t]
\caption{\(\textsc{ConInv}(\mathcal{D}^{\mathrm{grow}},\mathcal{X}_S)\): tentative shield construction}
\label{alg:coninv}
\begin{algorithmic}[1]
{\small
\REQUIRE{Accumulated grow dataset \(\mathcal{D}^{\mathrm{grow}}\), safe set \(\mathcal{X}_S\)}
\ENSURE{Tentative shield set \(\Omega^{\mathrm{tent}}\) and grow-phase safe action map \(\mu^{S,\mathrm{grow}}\)}

\STATE If \(\mathcal{X}\) is continuous, choose \(\delta_x\), construct the lattice \([\mathcal{X}_S]_{\delta_x}\), and evaluate the backward recursion on the lattice; see \cref{theorem:abstraction}
\STATE Initialize \(\Omega^{(0)}\leftarrow \mathcal{X}_S\)
\REPEAT
  \STATE Compute a conservative operator approximation
  \(\widetilde{\mathcal{Q}}_\epsilon^N(\Omega^{(\ell)})\)
  (or \(\widetilde{\mathcal{Q}}_{\delta_x,\epsilon}^N(\Omega^{(\ell)})\) in the lattice case)
  from \(\mathcal{D}^{\mathrm{grow}}\) using the backward recursion of \cref{theorem:Qtilde,theorem:abstraction}
  \STATE Update \(\Omega^{(\ell+1)}\leftarrow \widetilde{\mathcal{Q}}_\epsilon^N(\Omega^{(\ell)})\) (or its lattice counterpart)
\UNTIL{\(\Omega^{(\ell+1)}=\Omega^{(\ell)}\)}
\STATE Define \(\Omega^{\mathrm{tent}}\leftarrow \Omega^{(\ell)}\)
\STATE Extract the accompanying grow-phase safe action map \(\mu^{S,\mathrm{grow}}\) from the final backward recursion
\STATE Return \((\Omega^{\mathrm{tent}},\mu^{S,\mathrm{grow}})\)
}
\end{algorithmic}
\end{algorithm}

Algorithm~\ref{alg:certify_shield} describes the certification step.
Unlike \(\textsc{ConInv}\), this routine does \emph{not} run a fixed-point search.
Instead, it evaluates the conservative operator only once on the fixed tentative reference set \(\Omega^{\mathrm{tent}}\) using the certification data.
If the inclusion test
\(
\Omega^{\mathrm{tent}}\subseteq \widetilde{\mathcal{Q}}_{\epsilon}^{N,\mathrm{cert}}(\Omega^{\mathrm{tent}})
\)
passes, then, by the sample-splitting certification result of the previous section,
the tentative set is an \((N,\epsilon)\)-PCIS with confidence at least \(\eta\).
The safe action map deployed after acceptance is the one extracted from this certification recursion.

\renewcommand{\algorithmicrequire}{\textbf{Input:}}
\renewcommand{\algorithmicensure}{\textbf{Output:}}
\begin{algorithm}[t]
\caption{\(\textsc{CertifyShield}(\mathcal{D}^{\mathrm{cert}},\Omega^{\mathrm{tent}})\): hold-out certification of a tentative shield}
\label{alg:certify_shield}
\begin{algorithmic}[1]
{\small
\REQUIRE{Certification dataset \(\mathcal{D}^{\mathrm{cert}}\), tentative shield set \(\Omega^{\mathrm{tent}}\subseteq\mathcal{X}_S\)}
\ENSURE{Acceptance flag \(a\in\{0,1\}\) and certification safe action map \(\mu^{S,\mathrm{cert}}\)}

\STATE If \(\mathcal{X}\) is continuous, choose \(\delta_x\), construct the lattice representation of \(\Omega^{\mathrm{tent}}\), and evaluate the backward recursion on the corresponding lattice
\STATE Compute the certification operator evaluation
\[
\Omega^{\mathrm{cert}}
\leftarrow
\widetilde{\mathcal{Q}}_{\epsilon}^{N,\mathrm{cert}}(\Omega^{\mathrm{tent}})
\]
(or its lattice counterpart)
from \(\mathcal{D}^{\mathrm{cert}}\) using the backward recursion of \cref{theorem:Qtilde,theorem:abstraction} with the \emph{fixed} reference set \(\Omega^{\mathrm{tent}}\)
\STATE Extract the accompanying certification safe action map \(\mu^{S,\mathrm{cert}}\) from the same backward recursion
\IF{\(\Omega^{\mathrm{tent}}\subseteq \Omega^{\mathrm{cert}}\)}
  \STATE Set \(a\leftarrow 1\)
\ELSE
  \STATE Set \(a\leftarrow 0\)
\ENDIF
\STATE Return \((a,\mu^{S,\mathrm{cert}})\)
}
\end{algorithmic}
\end{algorithm}

\begin{remark}[How to enforce independence in practice]
\label{rem:shield_independence}
\normalfont
For the sample-splitting certification argument to apply literally, the grow and certification
datasets should be generated by independent data-collection procedures.
This can be implemented either sequentially or in parallel.
In particular, it is acceptable to run a separate certification episode (or a separate certification
worker) in parallel with the grow rollout, provided that the certification stream is generated from
an independent environment instance, with fresh resets and an independent random seed, under a
fixed certification policy \(S_{\mathrm{cert}}\).
The resulting certification data must not be used for RL updates or for the grow-phase fixed-point
search.
What matters is probabilistic separation rather than temporal separation. Thus, simply labeling one of two interacting rollouts as ``certification'' is not sufficient if the
certification policy is adapted using grow-phase quantities (such as the current learner policy,
current shield, or replay buffer contents) during data collection.
In that case, the simple independence argument no longer applies directly.
If certification is repeated at multiple shield updates, one should use either a fresh certification
batch at each update or a family of pre-allocated disjoint hold-out batches.
Repeated reuse of the same certification data across adaptively selected tentative shields would
require an additional multiple-testing or uniform-certification argument.
\end{remark}

\begin{remark}[Monotone shield update]
\label{rem:certified_monotone}
\normalfont
The inclusion test
\[
\widehat{\Omega}_i\subseteq \Omega_i^{\mathrm{tent}}
\]
in Algorithms~\ref{alg:generic_shield} and \ref{alg:shielded_dqn} is an optional pragmatic safeguard
that prevents the accepted shield from shrinking because of transient estimation noise.
If such monotonicity is not desired, one may update the current shield to
\((\Omega_i^{\mathrm{tent}},\mu_i^{S,\mathrm{cert}})\) whenever the certification step accepts.
\end{remark}


\section{Experiments}\label{sec:experiments}

We present a proof-of-concept study on a modified MountainCar benchmark.
The objective of this section is to evaluate whether the proposed PCIS-based shield can suppress unsafe exploration while preserving learning performance.

\subsection{Simulation settings}
We consider a modified MountainCar environment with continuous state
\(
\bm{s}=(x,v)\in\mathbb{R}^2
\),
where \(x\) and \(v\) denote position and velocity, and discrete action set
\(
\mathcal{U}=\{0,1,2\}
\)
corresponding to left thrust, no thrust, and right thrust.
The dynamics are
\begin{align}
v_{t+1} &= v_t + 10^{-3}(u_t-1) - 2.5\times 10^{-3}\cos(3x_t), \notag\\
x_{t+1} &= x_t + v_{t+1}.
\label{eq:mc_dynamics}
\end{align}
Unlike the standard clipped MountainCar benchmark, the usual environment-side boundary clipping/reset is disabled here so that safety violations are directly observable as excursions outside the prescribed safe set.
In other words, \eqref{eq:mc_dynamics} should be interpreted as the dynamics of an unclipped MountainCar variant rather than those of the standard benchmark.
The goal region is reached when \(x_t\ge 0.5\) and \(v_t\ge 0\), and the reward is \(r_t=-1\) at non-goal states and \(r_t=0\) at the goal.
The safe set is defined as
\begin{equation}
\mathcal{X}_S
=
[-1.5,0.6]\times[-0.07,0.07].
\label{eq:safe_set_exp}
\end{equation}
The initial state distribution follows the standard MountainCar initialization,
\(
x_0\sim U([-0.6,-0.4])
\)
with \(v_0=0\), although the subsequent dynamics are those of the modified unclipped variant described above. We discretize \(\mathcal{X}_S\) using a Cartesian grid with \(200\) position points and \(30\) velocity points.
Hence, the lattice spacings are approximately
\(
\Delta_x = 2.1/199 \approx 1.06\times 10^{-2}
\)
and
\(
\Delta_v = 0.14/29 \approx 4.83\times 10^{-3}
\).
Define
\begin{align}
\bar x = \frac{x+1.5}{2.1},\ 
\bar v = \frac{v+0.07}{0.14},
\label{eq:state_norm}
\end{align}
and let \(\bar{\bm{s}}=[\bar x,\bar v]^\top\).
With
\(
\mathcal{C}=\{0,1,2,3,4,5\}^2
\),
the state-only Fourier feature is
\begin{equation}
\psi(\bm{s})
=
\bigl[\cos(\pi\,\bm{c}^\top \bar{\bm{s}})\bigr]_{\bm{c}\in\mathcal{C}}
\in \mathbb{R}^{|\mathcal{C}|}.
\label{eq:fourier_basis}
\end{equation}
Fourier features are a standard choice for linear function approximation in reinforcement learning (see, e.g., \cite{lucas}).
Let \(e_u\in\{0,1\}^{|\mathcal{U}|}\) denote the one-hot vector corresponding to action \(u\in\mathcal{U}\).
The state-action feature used in the PCIS computation is then defined by
\begin{equation}
\phi(\bm{s},u)
=
e_u \otimes \psi(\bm{s})
\in
\mathbb{R}^{|\mathcal{U}||\mathcal{C}|},
\label{eq:sa_feature}
\end{equation}
where \(\otimes\) denotes the Kronecker product.
Equivalently, for each action \(a\in\mathcal{U}\) and Fourier index \(\bm{c}\in\mathcal{C}\),
\begin{equation}
\phi_{a,\bm{c}}(\bm{s},u)
=
\mathbf{1}\{a=u\}\cos(\pi\,\bm{c}^\top \bar{\bm{s}}).
\label{eq:sa_feature_component}
\end{equation}
Thus, only the block corresponding to the executed action is populated by the Fourier basis values, while the other action blocks are zero. In the present setting, \(|\mathcal{U}|=3\) and \(|\mathcal{C}|=6^2=36\), so the resulting state-action feature dimension is \(3\times 36=108\).

For DQN, the proposal policy is represented by an MLP with architecture
\(2\text{--}128\text{--}128\text{--}3\),
ReLU activations, Adam with learning rate \(10^{-4}\), discount factor \(\gamma=0.99\), replay size \(10^5\), mini-batch size \(64\), and target-network synchronization every \(10\) steps.
Exploration follows the exponentially decaying schedule
\begin{align}
&\varepsilon(t)
=
\varepsilon_{\min}
+
(\varepsilon_{\max}-\varepsilon_{\min})e^{-t/\tau},\notag \\
&\varepsilon_{\max}=1.0,
\ \varepsilon_{\min}=0.01,
\ \tau=1000.
\label{eq:eps_schedule}
\end{align}
To isolate the effect of shielding, the DQN comparison uses the same backbone and the same exploration schedule with the shield in \eqref{eq:shield_filter} either enabled or disabled.
In the DQN experiments, one update interval consists of \(150\) steps, and the total number of environment steps is \(10{,}000\).

We also tested a shielded SARSA agent using the same PCIS module and the same shielding mechanism.
Its proposal policy is a linear true-online SARSA(\(\lambda\)) agent with \(\alpha=10^{-3}\), \(\gamma=0.99\), \(\lambda=0.9\), and a linearly decaying exploration parameter from \(0.5\) to \(0.01\).
In the shielded SARSA experiments, one update interval consists of \(300\) executed steps and the maximum number of update intervals is \(100\). Here, \(\alpha_{\eta}\) denotes the implementation-level tuning parameter used in the confidence term of the PCIS update; it is distinct from the theoretical confidence levels \(\eta\) and \(\{\delta_j\}\).

\begin{table}[t]
\caption{DQN hyperparameters used for the shield on/off comparison}
\label{tab:exp_dqn}
\centering
\footnotesize
\setlength{\tabcolsep}{4pt}
\begin{tabular}{ll}
\toprule
Q-network & MLP \(2\text{--}128\text{--}128\text{--}3\) \\
Activation & ReLU \\
Optimizer & Adam \\
Learning rate & \(10^{-4}\) \\
Discount factor & \(\gamma=0.99\) \\
Replay size & \(10^5\) \\
Mini-batch size & \(B=64\) \\
Target update & every \(K=10\) steps \\
\(\varepsilon\)-greedy & \(\varepsilon_{\max}=1.0\), \(\varepsilon_{\min}=0.01\) \\
Decay scale & \(\tau=1000\) \\
Update-interval length & 150 steps \\
DQN step budget & \(10{,}000\) \\
\bottomrule
\end{tabular}
\end{table}

\begin{figure*}[t]
\centering
\includegraphics[width=1\linewidth]{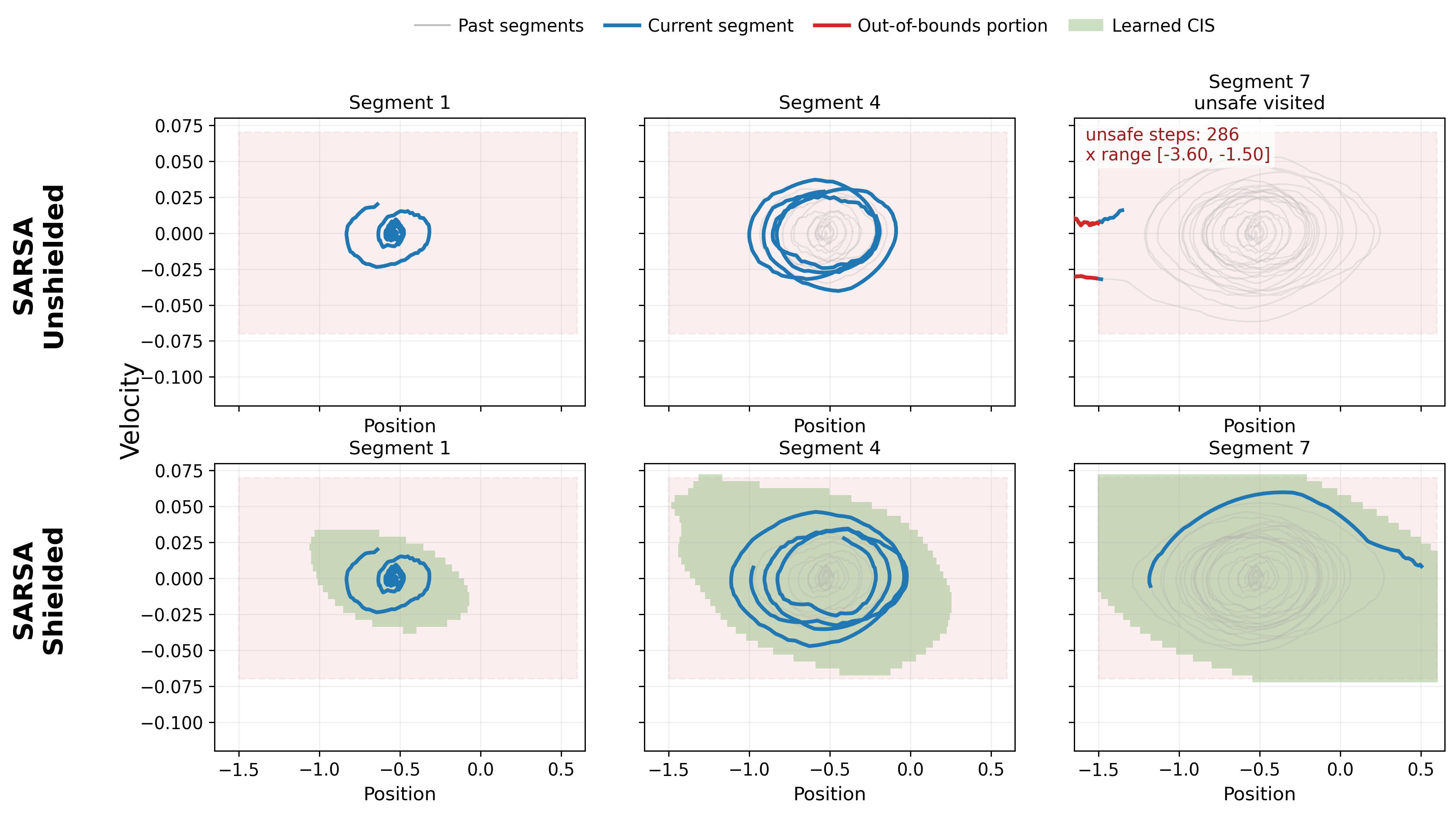}
\caption{Representative state-space trajectories for unshielded SARSA and shielded SARSA. The bold blue curve denotes the trajectory portion generated in the current update interval, the thin gray curves denote earlier update intervals, the red dashed rectangle indicates the safe set, and the green region indicates the current candidate PCIS  when shielding is active.}
\label{fig:results}
\end{figure*}

\begin{figure*}[t]
\centering
\includegraphics[width=1\linewidth]{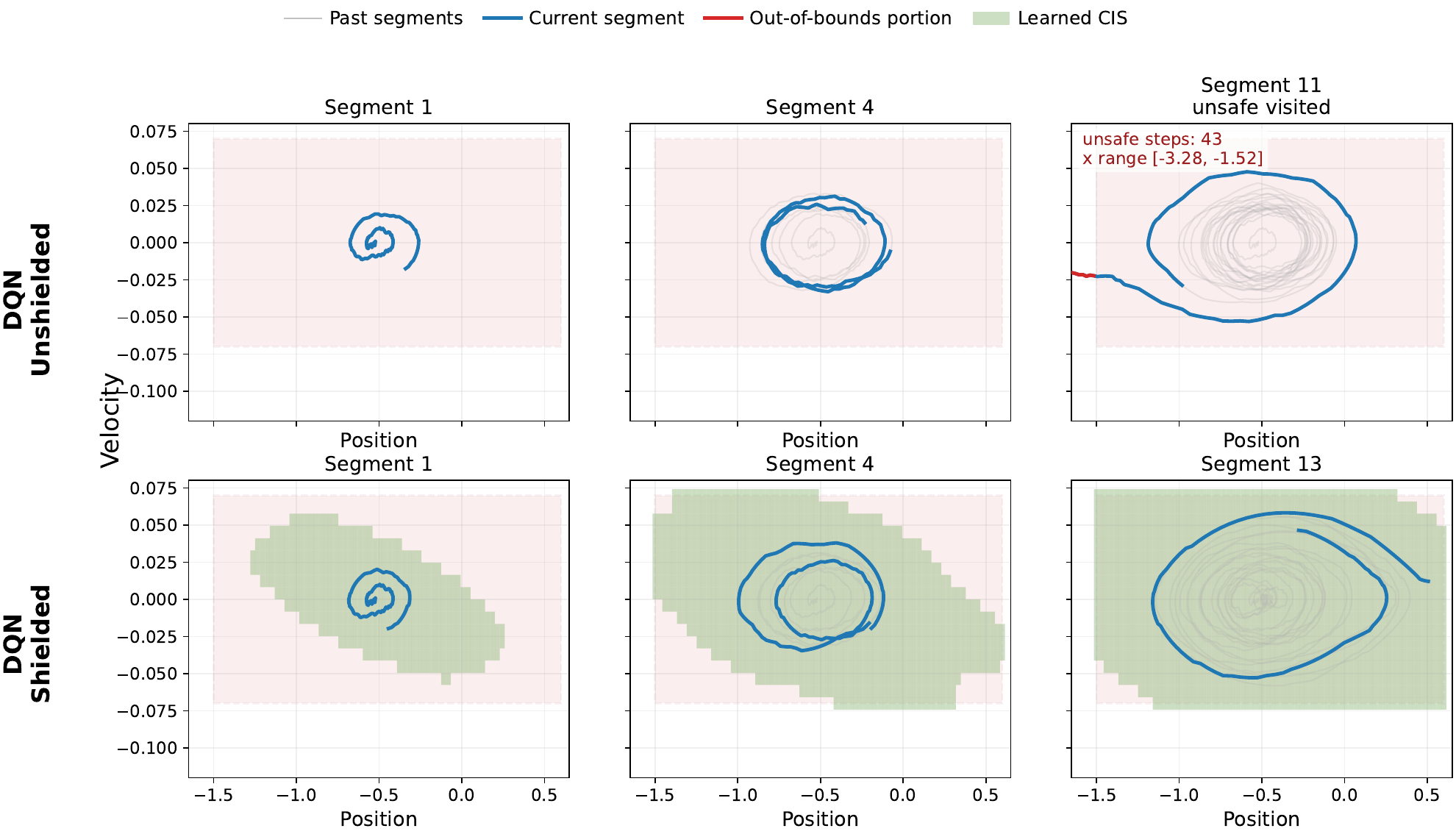}
\caption{Representative state-space trajectories for unshielded DQN and shielded DQN. The bold blue curve denotes the trajectory portion generated in the current update interval, the thin gray curves denote earlier update intervals, the red dashed rectangle indicates the safe set, and the green region indicates the current candidate PCIS  when shielding is active.}
\label{fig:results2}
\end{figure*}

\begin{figure}[t]
\centering
\includegraphics[width=\linewidth]{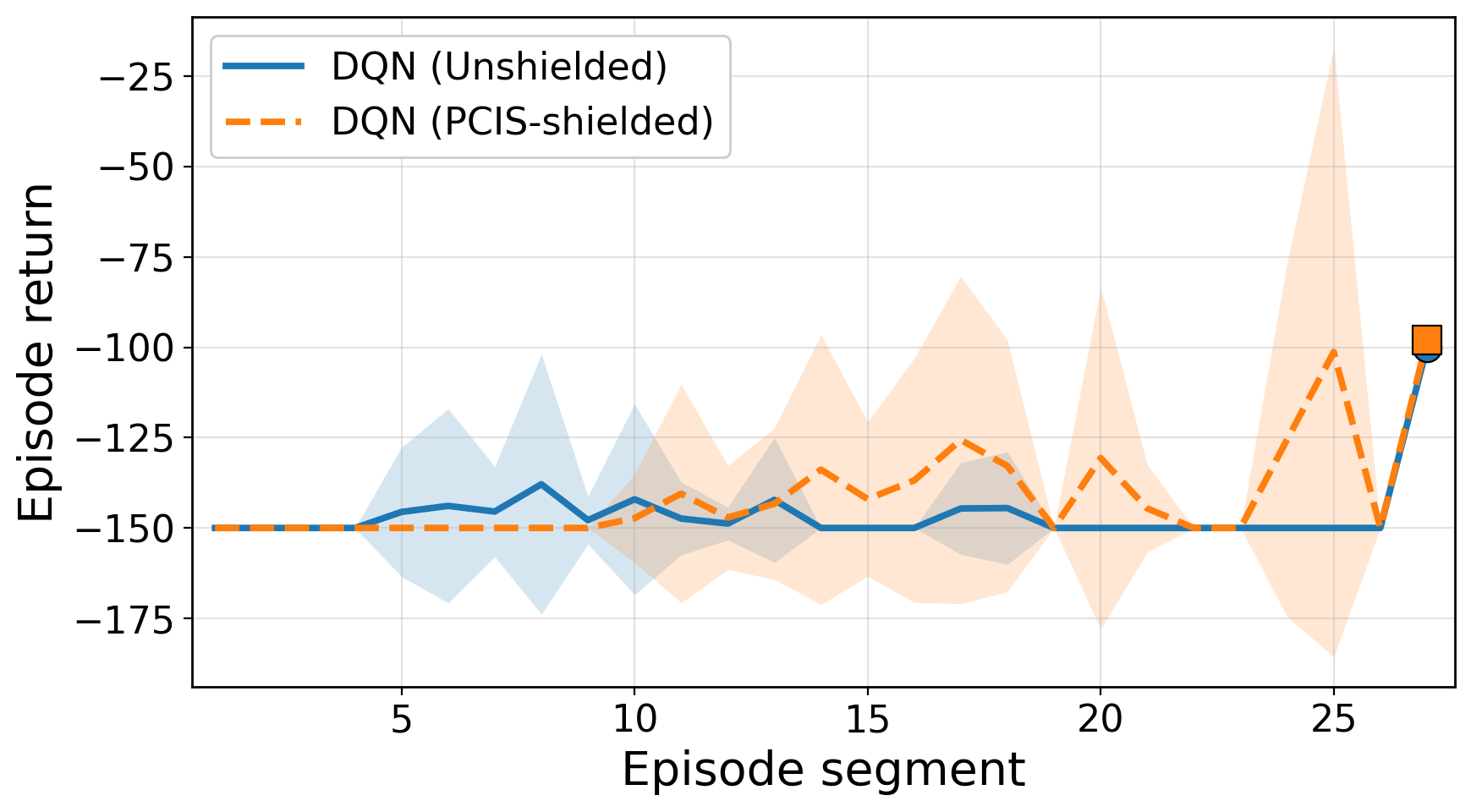}
\caption{Mean \(\pm\) standard deviation of returns accumulated over update intervals for DQN with and without shielding, over 30 random seeds.}
\label{fig:return_compare}
\end{figure}

\begin{figure}[t]
\centering
\includegraphics[width=\linewidth]{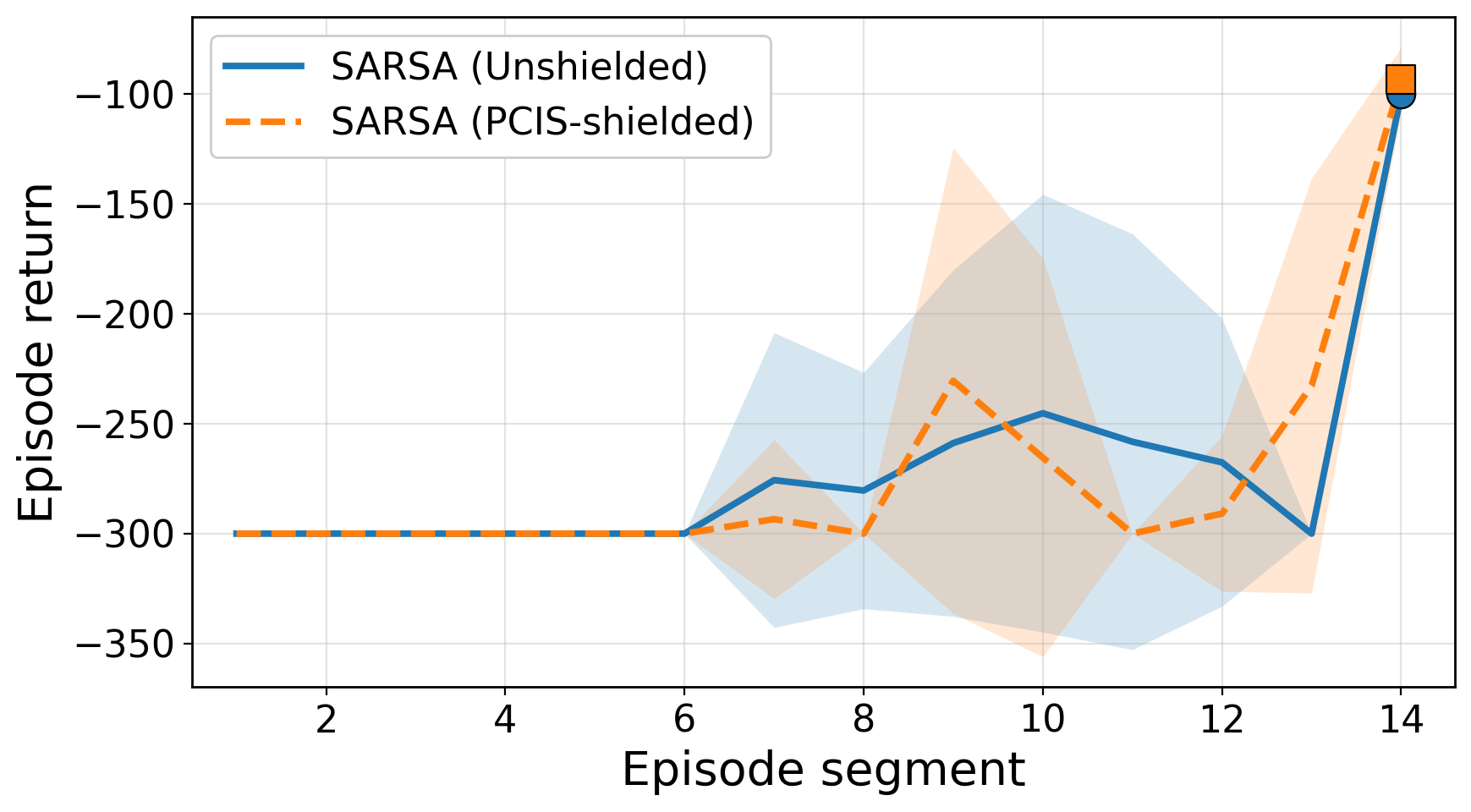}
\caption{Mean $\pm$ standard deviation of returns accumulated over update intervals for SARSA with and without shielding, over 30 random seeds.}
\label{fig:return_compare_sarsa}
\end{figure}

\begin{figure}[t]
\centering
\includegraphics[width=\linewidth]{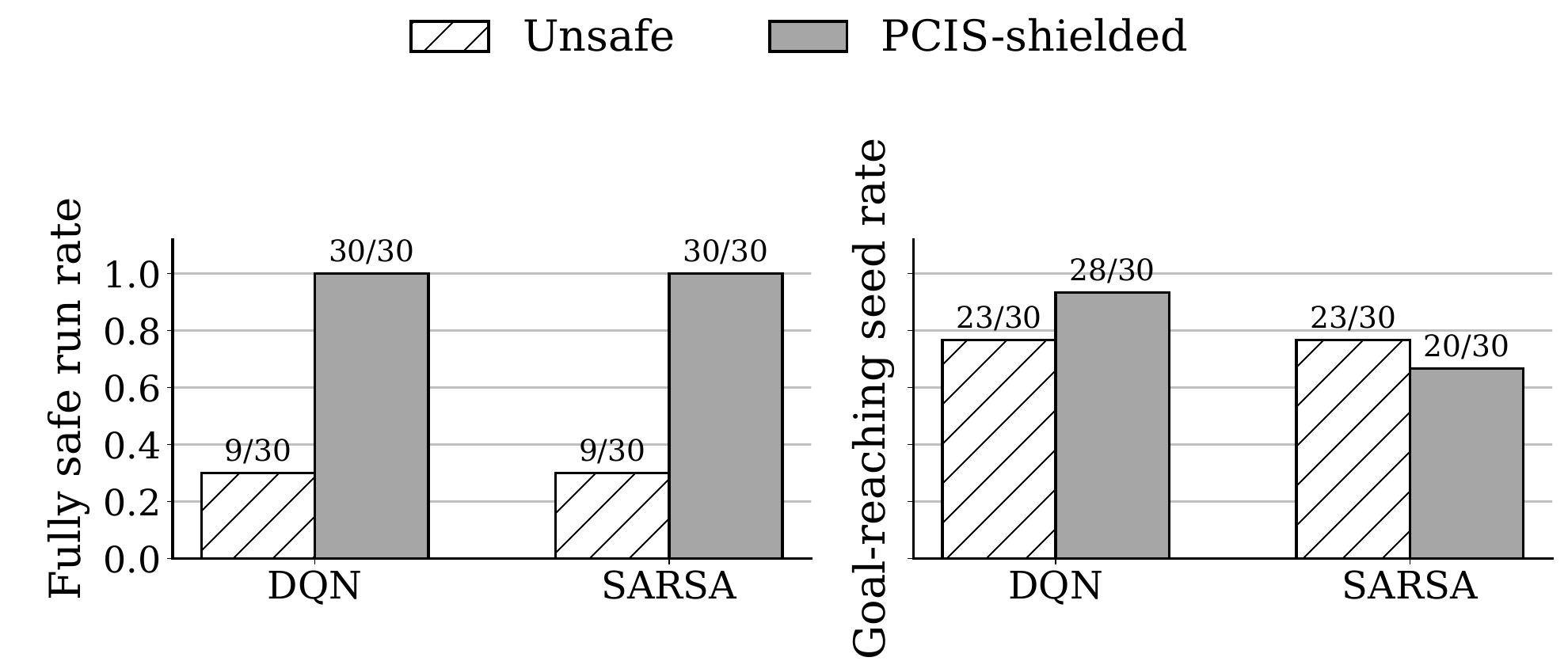}
\caption{Comparison of fully safe run rates (left) and goal-reaching seed rates (right) over 30 random seeds on the modified MountainCar environment. A run is counted as fully safe if it incurs zero unsafe steps over the entire 4000-step training trajectory, and a seed is counted as goal-reaching if the agent reaches the goal at least once during that training budget. The shielded results correspond to the final safety-priority configurations for DQN and SARSA under the same backbone hyperparameters and the same \(\varepsilon\)-greedy base policy as in the corresponding unshielded runs.}
\label{fig:safety_summary}
\end{figure}
\subsection{Simulation results}
Figures~\ref{fig:results} and \ref{fig:results2} show selected representative state-space trajectories for SARSA and DQN, respectively.
In each panel, the red dashed rectangle indicates the prescribed safe set \(\mathcal{X}_S\), the green region indicates the current candidate PCIS  displayed for the current update interval when shielding is active, the bold blue curve denotes the trajectory portion generated in the current update interval, and the thin gray curves denote trajectory portions accumulated in earlier update intervals.
The unshielded panels exhibit excursions outside \(\mathcal{X}_S\), whereas the shielded panels remain visually concentrated within the current candidate PCIS  once shielding becomes active.
Because runs terminate upon reaching the goal, the number of displayed update intervals differs across methods.
Moreover, the DQN and SARSA update intervals have different lengths, so the absolute values of update-interval returns are not directly comparable across the two algorithms.
These trajectory plots should therefore be interpreted as qualitative illustrations of the closed-loop behavior rather than as statistical evidence.

Figures~\ref{fig:return_compare} and \ref{fig:return_compare_sarsa} report the returns accumulated over update intervals as mean \(\pm\) standard deviation over 30 random seeds.
For both DQN and SARSA, the shielded and unshielded return curves remain on the same qualitative scale over much of the displayed horizon, suggesting that shielding does not dramatically degrade learning performance in this proof-of-concept setting.
At the same time, the variability across seeds is non-negligible, and the later portions of the curves should be interpreted with care because different runs terminate after different numbers of update intervals, so fewer seeds may contribute there.
For SARSA, the shielded curve appears somewhat more conservative in parts of the horizon, but it remains broadly comparable to the unshielded curve in overall scale.
These return plots should therefore be interpreted primarily as qualitative evidence that the shield does not induce a large performance deterioration, rather than as a fine-grained statistical comparison of learning efficiency.

Figure~\ref{fig:safety_summary} summarizes two run-level metrics over 30 random seeds: the fully safe run rate and the goal-reaching seed rate.
A run is counted as fully safe if it incurs zero unsafe steps over the entire 4000-step training trajectory, and a seed is counted as goal-reaching if the agent reaches the goal at least once during that training budget.
Under the fully safe run metric, DQN improves from \(9/30\) (\(30.0\%\)) fully safe runs without shielding to \(30/30\) (\(100\%\)) with shielding, while the goal-reaching seed rate increases from \(23/30\) (\(76.7\%\)) to \(28/30\) (\(93.3\%\)).
For SARSA, shielding likewise improves the fully safe run rate from \(9/30\) (\(30.0\%\)) to \(30/30\) (\(100\%\)), whereas the goal-reaching seed rate changes from \(23/30\) (\(76.7\%\)) without shielding to \(20/30\) (\(66.7\%\)) with shielding.
This modest reduction for SARSA is consistent with the more conservative safety-priority shield: in MountainCar, reaching the goal requires sufficiently aggressive momentum-building motion, and restricting actions near the boundary of the candidate safe region can make such trajectories less likely.
These results indicate that, on this modified MountainCar benchmark, the PCIS-based shield can substantially improve safety under a fair comparison between the shielded and unshielded settings.
For DQN, this safety improvement is accompanied by a higher goal-reaching seed rate, whereas for SARSA the safety-priority shield yields a modest reduction in goal-reaching frequency.



\section{Conclusion and Future Directions}
\label{sec:limitations}
The main contribution of this paper is to show that probabilistic controlled invariance can be synthesized directly from data for unknown linear MDPs, yielding an $\eta$-conservative safe set and an accompanying set-valued shield. The present formulation intentionally focuses on the regime in which such guarantees are currently tractable: finite actions and either discrete states or low-dimensional continuous states equipped with a lattice abstraction. Within this regime, the method provides a transparent pipeline from transition data to a deployable runtime safety filter, which can be coupled with standard RL learners without modifying their internal update rules.

Several extensions appear promising. First, the scalability issue of the lattice abstraction is not unique to our setting; rather, it reflects the usual state-explosion phenomenon in abstraction-based synthesis. A natural next step is therefore to combine the present PCIS recursion with compositional or subsystem-based abstractions, so that local certificates are computed on lower-dimensional components and then assembled into a global safety certificate \cite{lavaei2020compositional}. Likewise, adaptive or nonuniform discretizations could concentrate computation near the boundary of the safe set or in regions where the feature map varies rapidly, thereby reducing unnecessary refinement elsewhere. Second, the dependence on a global Lipschitz constant $L_\phi$ should be viewed as a conservative but convenient first step rather than as a fundamental obstacle. In practice, one may expect tighter implementations based on local Lipschitz estimates, data-dependent abstraction errors, or learned state-space metrics, which could reduce conservatism while preserving the same conservative-inclusion logic.
Third, the finite-action assumption mainly enters through the maximization and safe-action selection steps. This suggests a concrete path toward large or continuous action spaces: exhaustive enumeration could be replaced by an optimization oracle over a compact action set, or by adaptive candidate sets constructed from coverings, local linearization, or feature-space maximization. This direction is consistent with existing results on linear bandits with possibly infinite action sets \cite{abbasi2011online}, recent work on linear MDPs with large action spaces \cite{xu2023largeaction}, and continuous-action extensions for low-rank MDPs under additional smoothness assumptions \cite{oprescu2024continuous}.
Overall, we view the present method as a certifiable baseline rather than an endpoint: it isolates a practically meaningful setting in which data-driven PCIS synthesis and shielding can be established cleanly, while also indicating concrete routes toward higher-dimensional systems and richer action spaces. Future work will pursue these directions together with sharper confidence accounting for repeated online shield updates.

\bibliographystyle{IEEEtran}
\bibliography{myrefs_v11_ref30_neurips_and_future_directions}

\appendix
\section{Proof of Theorem~2}

\begin{proof}
For each \(j\), because \(\bar p_{j+1}^{\delta_x}\) is computed from later stages only, it is independent of the dataset \(\mathcal{D}^{(j)}\).
By backward induction from
\eqref{eq:ptilde_dx_terminal}, \eqref{eq:ptilde_dx_rec},
\eqref{eq:pbar_dx_terminal}, and \eqref{eq:pbar_dx_rec},
\begin{equation}
0 \le \tilde p_j^{\delta_x}(x_d) \le 1,\quad
0 \le \bar p_j^{\delta_x}(x) \le 1
\label{eq:pbar_dx_bounded}
\end{equation}
for all \(x_d\in[\mathcal{X}_S]_{\delta_x}\), \(x\in\mathcal{X}\), and \(j=0,\ldots,N\).
Conditioned on \(\bar p_{j+1}^{\delta_x}\), \cref{lemma:expectation} implies that there exists
\(
\theta_j^{\delta_x}\in\mathbb{R}^d
\)
such that
\[
\mathbb{E}[\bar p_{j+1}^{\delta_x}(x')\mid x,u]
=
(\theta_j^{\delta_x})^\top \phi(x,u)
\]
for all \((x,u)\).
Moreover, \(\|\theta_j^{\delta_x}\|_\infty\le 1\), hence \(\|\theta_j^{\delta_x}\|_1\le d\).
Let
\(
\mathcal{E}^{\delta_x}=\bigcap_{j=0}^{N-1}\mathcal{E}_j^{\delta_x}
\).
By \eqref{eq:Ej_dx_prob} and the union bound,
\[
\Pr(\mathcal{E}^{\delta_x})
\ge
1-\sum_{j=0}^{N-1}\delta_j
\ge \eta.
\]

We now prove by backward induction that, on \(\mathcal{E}^{\delta_x}\),
\begin{equation}
\bar p_j^{\delta_x}(x)\le p_j^\Omega(x),\ 
\forall x\in\mathcal{X},\ \forall j=0,\ldots,N.
\label{eq:induction_goal_dx}
\end{equation}
The claim is immediate at \(j=N\).
Indeed, if \(x\notin\Omega\), then
\(
\bar p_N^{\delta_x}(x)=0=p_N^\Omega(x)
\),
while for \(x\in\Omega\),
\[
\bar p_N^{\delta_x}(x)
=
\1_\Omega(q_{\delta_x}(x))
\le 1
=
p_N^\Omega(x).
\]
Assume \eqref{eq:induction_goal_dx} holds at stage \(j+1\).
Fix \(x\in\mathcal{X}\).
If \(x\notin\Omega\), then
\(
\bar p_j^{\delta_x}(x)=0=p_j^\Omega(x)
\).
Now let \(x\in\Omega\) and set \(x_d=q_{\delta_x}(x)\).
Using \eqref{eq:phi_lipschitz}, \eqref{eq:delta_net}, and \(\|\theta_j^{\delta_x}\|_1\le d\),
\[
\left|
(\theta_j^{\delta_x})^\top\phi(x,u)
-
(\theta_j^{\delta_x})^\top\phi(x_d,u)
\right|
\le
dL_\phi\delta_x.
\]
On \(\mathcal{E}_j^{\delta_x}\),
\[
(\theta_j^{\delta_x})^\top\phi(x_d,u)
\ge
(\hat\theta_j^{\delta_x})^\top\phi(x_d,u)-\beta_j\sigma_j(x_d,u).
\]
Hence, for any \(u\in\mathcal{U}\),
\begin{align}
\mathbb{E}[\bar p_{j+1}^{\delta_x}(x')\mid x,u]
&=
(\theta_j^{\delta_x})^\top\phi(x,u) \notag \\
&\ge
(\hat\theta_j^{\delta_x})^\top\phi(x_d,u)
-dL_\phi\delta_x
-\beta_j\sigma_j(x_d,u) \notag \\
&=
\ell_j^{\delta_x}(x_d,u).
\label{eq:barp_lower_dx}
\end{align}
By the induction hypothesis,
\[
\mathbb{E}[p_{j+1}^\Omega(x')\mid x,u]
\ge
\mathbb{E}[\bar p_{j+1}^{\delta_x}(x')\mid x,u]
\ge
\ell_j^{\delta_x}(x_d,u).
\]
Therefore,
\begin{align}
p_j^\Omega(x)
&=
\max_{u\in\mathcal{U}}
\mathbb{E}[p_{j+1}^\Omega(x')\mid x,u] \notag\\
&\ge
\max_{u\in\mathcal{U}}
\sat{\ell_j^{\delta_x}(x_d,u)}.
\end{align}
If \(x_d\notin\Omega\), then
\(
\tilde p_j^{\delta_x}(x_d)=0
\).
Otherwise,
\[
\tilde p_j^{\delta_x}(x_d)
=
\max_{u\in\mathcal{U}}
\sat{\ell_j^{\delta_x}(x_d,u)}
\le p_j^\Omega(x).
\]
In either case,
\[
\bar p_j^{\delta_x}(x)
=
\tilde p_j^{\delta_x}(x_d)
\le p_j^\Omega(x),
\]
which proves \eqref{eq:induction_goal_dx}.
Consequently, on \(\mathcal{E}^{\delta_x}\),
\begin{align}
\widetilde{\mathcal{Q}}^N_{\delta_x,\epsilon}(\Omega)
&=
\{x\in\Omega:\bar p_0^{\delta_x}(x)\ge 1-\epsilon\} \notag\\
&\subseteq
\{x\in\Omega:p_0^\Omega(x)\ge 1-\epsilon\} \notag\\
&=
\mathcal{Q}^N_\epsilon(\Omega).
\end{align}
Since \(\Pr(\mathcal{E}^{\delta_x})\ge \eta\), \eqref{eq:Qtilde_dx_conservative} follows.
\end{proof}

\end{document}